\documentclass[12pt]{article}
\RequirePackage[OT1]{fontenc}
\usepackage{amsfonts,amsmath,amsthm,amssymb,bm,bbm}
\usepackage{graphicx}
\usepackage{array}
\usepackage{natbib}
\usepackage{enumerate}
\usepackage{booktabs}
\usepackage{parskip}
\usepackage[colorlinks=true,pagebackref,linkcolor=magenta,citecolor=blue]{hyperref}
\usepackage{multirow}
\usepackage{tikz}
\usepackage{verbatim}
\usepackage[top=1in, bottom=1in, left=1in, right=1in]{geometry}
\usepackage{setspace}
\usepackage{array}
\usepackage{rotating}
\usepackage{float}
\usepackage{pgf}
\usetikzlibrary{arrows,shapes,positioning}

\usepackage{subfigure}
\RequirePackage{graphicx}
\RequirePackage{xspace}
\RequirePackage{tabularx}
\RequirePackage{comment}
\RequirePackage{subdepth}

\usepackage{fancyhdr}
\newtheorem{theorem}{Theorem}[section]
\newtheorem{lemma}[theorem]{Lemma}

\newtheorem{proposition}[theorem]{Proposition}

\theoremstyle{definition}
\newtheorem{definition}[theorem]{Definition}
\newtheorem{remark}[theorem]{Remark}

\newtheorem{example}[theorem]{Example}
\numberwithin{equation}{section}

\renewcommand{\tilde}{\widetilde}
\renewcommand{\hat}{\widehat}
\renewcommand{\Pr}{\mathbb{P}}
\renewcommand{\Re}{\mathbb{R}}

\DeclareMathOperator{\Normal}{\mathcal{N}}

\DeclareMathOperator{\Ex}{\mathbb{E}}

\newcommand{\ER}{Erd{\H o}s-R{\' e}nyi}

\title{Elements of estimation theory for causal effects in the presence of network interference}

\author{Daniel L. Sussman\footnote{Boston University, Department of Mathematics \& Statistics,  Boston, MA, sussman@bu.edu}, Edoardo M. Airoldi\footnote{Harvard University, Department of Statistics, Cambridge, MA, airoldi@fas.harvard.edu}}

\begin{document}    

\maketitle

\begin{abstract}
Randomized experiments in which the treatment of a unit can affect the outcomes of other units are becoming increasingly common in healthcare, economics, and in the  social and information sciences.
From a causal inference perspective, the typical assumption of no interference becomes untenable in such experiments.
In many problems, however, the patterns of  interference may be informed by the observation of network connections among the units of analysis.
Here, we develop elements of optimal estimation theory for causal effects leveraging an observed network, by assuming that the potential outcomes of an individual depend only on the individual's treatment and on the treatment of the neighbors.
We propose a collection of exclusion restrictions on the potential outcomes, and show how subsets of these restrictions lead to various parameterizations.
Considering the class of linear unbiased estimators of the average direct treatment effect, we derive conditions on the design that lead to the existence of unbiased estimators, and offer analytical insights on the weights that lead to minimum integrated variance estimators.
We  illustrate the improved performance of these estimators when compared to more standard biased and unbiased estimators, using simulations.\\
{\bf Keywords}: causal inference, network data, interference, inferential targets, exclusion restrictions, general additive models, unbiased estimation, minimum integrated variance \end{abstract}


\section{Introduction}
Social scientists, businesses, public health researchers and economists frequently seek to perform experiments in environments where it may be prohibitive or undesirable to restrict interactions between individuals in the study \citep{gui2015network}.
Experiments in these contexts present challenges and opportunities, as the causal effects of treatments need not be confined to treated individual.
A core assumption in causal inference of no interference \citep{cox1958planning,rubin1974estimating,manski1990nonparametric} frequently cannot be made for such studies.

In the presence of interference, two key questions arise: 
(i) Which units' treatments can effect which units' outcomes? (ii) How can treatments effect each unit's outcomes?
Fortunately in some instances, the answer to the first question is partially exposed by an observed social network based on an online user base \citep{bond201261,gui2015network} or collected as part of a study \citep{mcqueen2015national}
Provided the network is an accurate depiction of which individuals interact with each other, a tempting and often sensible tact is for an experimenter to assume that the transmission of effects of treatments must respect the structure of the network.

For the second question, various properties of the treatments, outcomes and the social network itself may be used to justify further assumptions regarding the nature of how outcomes relate to treatments.
Furthermore, making additional assumptions may relax the set of reasonable analytic tools sufficiently to allow for variance reduction in exchange for a small increase in bias.

We consider a set of core assumptions for interference when an observed network illuminates the structure of this interference.
Combinations of these core assumptions lead to a variety of parametric forms for the potential outcomes which serve to decompose treatment effects as either direct, indirect or interaction effects. 
From a randomization inference perspective, we consider the set of unbiased estimates of the direct treatment effect, deriving conditions on the experimental design that ensure linear unbiased estimates exist.
We also propose a novel optimality criterion for this setting  which draws on Bayesian ideas, deriving minimum integrated variance unbiased estimates. 
We demonstrate that using appropriate distributions for the parameters leads to improved performance over more standard methods to create unbiased estimators as well as over naive biased estimates. 

In Section~\ref{sec:background}, we review causal inference with interference and networks and we introduce the potential outcome framework.
In Section~\ref{sec:nia}, we consider one core and four structurual assumptions regarding network interference and in Section~\ref{sec:model_estimand}, we show how these assumptions lead to parameterizations of the potential outcomes which clarify relevant estimands.
As many of the ideas can be translated between the various sets of assumptions, much of the focus of this manuscript is under a particular parmameterization where two of the four structural assumptions hold.
In Section~\ref{sec:unbiased} we consider the set of unbiased estimators of the direct treatment effect as controlled by collections linear equations and in Section~\ref{sec:mivlue} we propose a method to choose among this set of estimators by optimizing a quadratic objective with respect the linear constraints for unbiasedness.
Finally, in Section~\ref{sec:simulation} we perform simulations to illustrate the effectiveness of our methods and in Section~\ref{sec:disc} we discuss the ideas proposed and future directions.





\section{Background}
\label{sec:background}

The assumption of no interference is at the core of causal inference \citep{cox1958planning,manski1990nonparametric} as it is a key part of the stable unit treatment assumption \citep{rubin1974estimating}.
To relax this assumption, researchers have proposed different answers to the question of which units can interfere with which other units based on the type of experiment performed. 
Early forms of interference included interference localized to an individual across different rounds of treatment such as in clinical trials with crossover designs \citep{grizzle1965two}.
Interference based on spatial proximity of treated units also received relatively early attention \citep{kempton1984inter}.
More recently, interference confined to blocks where interference is allowed within blocks but not between blocks has been studied \citep{hudgens2012toward}.
General interference has been studied occasionally \citep{rosenbaum2007interference} but typically allowing for any unit's treatment to effect any unit's outcome makes design and analysis major challenges.

With the rise of interest in social networks, many authors have proposed that observed network of interactions or relationships between units may inform how causal effects can be transmitted between units \citep{ugander2013graph,eckles2014design,athey2015exact}.
An assumption that has gained traction for causal inference on networks is that the causal effects can be passed along edges in the network so that units may experience an effect from the treatment of other nodes that are in some ways proximal in the network.
Observation of the network itself does not determine which units interfere with each other but typically the network structure drives key assumptions about the interference.
One form of these assumptions is that the outcome of unit $i$ depends only on the treatments of units which can connected to $i$ following a path through a network of length at most $k$ for some $k$ \citep{athey2015exact}.
The strongest and most studied assumption is to assume that units outcomes are only impacted by their neighbors in the network, which will be the focus of this manuscript.

In the presence of network interference, the experimental design and analysis of the outcomes present distinct challenges.
Efforts to design experiments which takes the structure of interference into account include block and neighbor designs \citep{david1996designs}, graph cluster design \citep{ugander2013graph}, nomination designs, high-degree designs \citep{kim2015social}, peer encouragement designs \citep{eckles2016estimating} and model-based designs \citep{basse2015optimal}.
Many of these can be viewed as versions of classical blocking designs in the the network context while others are based off of how individuals communicate within the network.
Whether using a design which takes into account the network structure, or a standard design such as a Bernoulli trial, the subsequent analysis should take into account the structure of the interference.

Our goals are to build estimators that are unbiased and in some sense optimal for estimands under a variety of network interference assumptions. 
While unbiased estimation is not necessary for all scenarios and can often be traded off for a substantial variance reduction, understanding the structure of unbiased estimators can be useful on its own and can also illuminate how accepting small amounts of bias can lead to even better estimators.





\citet{aronow2013estimating} proposed a method to create unbiased estimators for average outcome for any given exposure to treatment under interference, where exposure levels are defined as a subset of treatment allocations for each unit.
As with Horvitz-Thompson estimation, these estimators rely on deriving the probability of a given exposure based on the design, known as the propoensity score, and using inverse propensity score weighting of the observed outcomes. 
While this method is generally applicable and easily computed, this is far from the only way to construct unbiased estimators and one of our goals is to investigate the questions, (i) when do unbiased estimators exist? (ii) what is the set of unbiased estimators for a given design, network, and interference assumptions? and (iii) how can we choose from among these estimators?

We concretely answer the first two questions and provide a criterion for answering the third. 
In particular, we propose finding the minimum integrated variance linear unbiased estimator with respect to a distribution on the potential outcomes.
We will call this distribution a prior distribution which can be viewed as an {\em a priori} weighting of the parameter values of interest.
Minimum integrated variance unbiased estimators have not been studied extensively in the causal inference or statistics literature.
One closely related area where these ideas have been investigated is in the field of finite population inference \citep{ghosh1997bayesian}.
Minimum integrated variance designs have also been studied in the design of experiments literature but these ideas seek a design that functions well for standard estimators, whereas are goal is the reverse \citep{box1959basis}.
In analogy to results from finite population statistics, we will see in Section~\ref{sec:mivlue} that estimators similar to those proposed by \citet{aronow2013estimating} do arise as minimum integrated variance solutions for certain prior distributions.
However, in our simulation studies (see Section~\ref{sec:simulation}) we show that other prior distributions on the parameters generally lead to improved performance compared to inverse propensity score weighting.

\subsection{Potential outcomes framework}

We will now give an overview of the potential outcomes framework for causal inference, sometimes referred to as the Rubin causal model \citep{neyman1923applications,rubin1974estimating}.
Denote the population of experimental units as $[n]=\{1,2,\dotsc, n\}$. 
A treatment allocation vector, or simply an allocation, is a vector $\mathbf{z}\in \mathcal{Z}^n=\{0,1\}^n$ where the $i^{\text{th}}$ component $z_i$ represents the treatment assigned to unit $i$. 
Our results extend in a straightforward way to finite treatment regimes, but as notation can become unwieldy, we consider the binary treatment regime.
The outcome for unit $i$ under allocation $\mathbf{z}$ is denoted as $Y_i(\bf{z})\in \Re$.

The design of a randomized trial is a probability distribution over all allocations with probability mass function denoted as $p:\mathcal{Z}^n\mapsto[0,1]$, also referred to as the randomization distribution.
%
In the experiment, a researcher will select an allocation $\mathbf{z}^{obs}$ distributed according to the design $p$.
The observed data consists of the allocation $\mathbf{z}^{obs}$ and the outcomes $Y_i^{obs}=Y_i(\mathbf{z}^{obs})$.
The support of the design is the set of allocations $\mathbf{z}\in \mathcal{Z}^n$ such that $p(\mathbf{z})>0$.
We take the view of randomization-based inference so that for any fixed $\mathbf{z}$ the potential outcomes for all units are non-random, meaning that all randomness in the observed data is due to the random selection of the treatment allocation $\mathbf{z}^{obs}$. 
Hence, we use $\Pr$ and $\Ex$ to denote probability and expectation with respect to the design $p$. 
Occasionally, it will be useful to refer to the allocation where only unit $j$ is in the treatment group and all others are in the control group; this allocation will be denoted as $\mathbf{e}_j$.

If one can reasonably assume that the experimental units are not interacting, it is plausible to make the stable unit treatment value assumption (SUTVA) \citep{rubin1974estimating}.
Under SUTVA, for each unit $i$ the outcome $Y_i(\bf{z})$ depends on $\bf{z}$ only through $z_i$ so we can abuse notation and write $Y_i(\bf{z})$ as $Y_i(z_i)$. 
Hence, SUTVA effectively restricts the set of potential outcomes for each unit to $Y_i(1)$ and $Y_i(0)$.
For any observed allocation $\mathbf{z}^{obs}$, one would observe $n$ out of the $2n$ distinct potential outcomes, leaving $n$ potential outcomes as unobserved.
Much of the work under the Rubin causal model can be viewed as developing and studying methods to deal with this missing data problem.

Without making any assumptions about interference, we must view the potential outcome for each individual and each treatment allocation as possibly depending on the treatment assignments for all individuals \citep{rubin1978bayesian}. 
In this situation, the set of distinct potential outcomes explodes in cardinality to $2^n$ for each individual and $n2^n$ for the experimental population.
We will still only observe $n$ of these potential outcomes and dealing with a missing data problem of this magnitude is unlikely to yield useful inferences about causal effects without further assumptions or restrictions.
However, we need not restrict ourselves to SUTVA or arbitrary interference, and in the next section we will explore a collection of assumptions for the interference between units in a network.



\section{Neighborhood Interference Assumptions}\label{sec:nia}

The key data for constraining interference between units in our setting is an observed network of relationships between units that might accounts for possible paths of interference.
We will first introduce some necessary notation for the network setting.

\subsection{Notation}
The network will be represented by a binary adjacency matrix $g=[g_{ij}]_{i,j=1}^{n}\in \{0,1\}^{n\times n}$.
The entry $g_{ij}$ is $1$ if there is a edge from unit $i$ to unit $j$ and $0$ if there is no edge.
For the purpose of identifiability, we assume that $g_{ii}=0$. 
For generality, the network is directed so that $g_{ij}$ need not equal $g_{ji}$, though all results below apply in the undirected case where $g_{ij}=g_{ji}$ for all $i,j$.
We will say that unit $j$ is a neighbor of unit $i$ if $g_{ji}=1$.

The neighborhood of unit $i$ is the set of vertices with edges directed towards unit $i$.
We denote the neighborhood by $\mathcal{N}_i=\{j: g_{ji}=1\}$. 
The degree of unit $i$ is denoted as $d_i=\sum_{j=1}^n g_{ji}=|\mathcal{N}_i|$.
For unit $i$ and allocation $\mathbf{z}$, we denote the vector of treatments for the neighbors of unit $i$ by $\mathbf{z}_{\mathcal{N}_i}=(z_j)_{j\in \mathcal{N}_i} \in \{0,1\}^{d_i}$.
The set of treated neighbors for unit $i$ is $\mathcal{N}_i^{\textbf{z}}=\{j:g_{ji}=1,z_j=1\}$.
We define the treated degree of unit $i$, denoted $d^{\mathbf{z}}_i$, to be the number of treated neighbors of unit $i$, $d_i^{\mathbf{z}}=|\mathcal{N}_i^{\textbf{z}}|=(g^T \mathbf{z})_i$.

\subsection{Core Assumption}
\label{sec:core_ass}

While many assumptions could be made with regard to how causal effects can propogate through the network, perhaps the most concrete assumption is the following.

\begin{definition}[Neighborhood Interference Assumption (NIA)]\label{def:nia}
For a network $g\in\{0,1\}^{n\times n}$, the potential outcomes satisfy the neighborhood interference assumption for $g$ if for each unit $i\in[n]$ and all treatment  allocations $\mathbf{z},\mathbf{z}'\in \mathcal{Z}$ where $z_j=z_j'$ for all $j\in \mathcal{N}_i \cup\{i\}$, it holds that $Y_i(\mathbf{z})= Y_i(\mathbf{z}')$ .
\end{definition}
NIA means that if the treatments of unit $i$ and its neighbors are held fixed, then changing the treatments of other units does not change the outcome for unit $i$.
Under NIA, we define an alternative function for the potential outcome in terms of the treatment of the unit $i$ and the treatment of its neighbors, which will help simplify exposition:
let the function $\tilde{Y}_i:\{0,1\}\times \{0,1\}^{d_i} \mapsto \Re$ satisfy $Y_i(\mathbf{z}) = \tilde{Y}_i(z_i, \mathbf{z}_{\mathcal{N}_i})$ for all allocations $\mathbf{z}$.

Depending on the structure of the observed network, NIA can be a significant reduction in the set of distinct potential outcomes as compared to an unconstrained interference setting.
In general, the number of distinct potential outcomes for each unit under NIA is at most $2^{d_i+1}$.
If the network is empty, so that $g_{ij}=0$ for all $i,j$, then NIA is equivalent to SUTVA while if the network is complete, so that $g_{ij}=1$ for all $i\neq j$, then NIA allows for arbitrary interference.
The results in the paper are best applied in an intermediate regime with sparse networks where $0<d_i\ll n$, which is the setting for many real-world networks.
Hence, conditioning on observing $g$ pre-treatment, we can limit the set of distinct potential outcomes in a way that allows for the development of estimation theory.

\begin{remark}
The network $g$ does not need to correspond to some observed interactions between units. 
It is easy to apply NIA in settings where $\mathcal{N}_i$ is defined in other ways; for example, the neighbors could correspond to nearby units in space or units sharing certain properties.
Importantly, the network must be observed or at least observable before the experiment and the network must be fixed and assumed not to be effected by treatment.\qed
\end{remark}

\subsection{Structural Assumptions}\label{sec:structural}
While assuming NIA answers the question of which units' treatments can effect which units' outcomes, the question of how the outcomes are effected remains open.
In this section, we consider four structural assumptions assumptions that will more clearly restrict how outcomes are effected by a neighbors' treatments.
These can be combined in various ways to fit different goals of the practitioner to restrict the set of potential outcomes.
Two of the assumptions are related to symmetry of the interference effects, and the other assumptions are related to additivity of effects.

Under no interference, the potential must satisfy the form $Y_i(z_i)=\alpha_i+\beta_i z_i$ which provides a relationship between the potential outcomes framework and linear regression.
Similarly, Each combination of these assumptions implies a parametric form for the potential outcomes. 
We will illustrate the parameterizations in Section~\ref{sec:param} for some combinations and the details for the remaining combinations are in Appendix~\ref{app:param}.
For each of the following definitions we assume that NIA holds.

\subsubsection{Additivity of Main Effects}
The first assumption is that there is no interaction between a unit's treatment and the treatments of a unit's neighbors.
\begin{definition}[Additivity of Main Effects]\label{def:add}
The potential outcomes satisfy additivity of main effects if  
\begin{equation}\label{eq:add_main}
  \tilde{Y}_i(z_i,\mathbf{z}_{\mathcal{N}_i}) = \tilde{Y}_i(0,\bm{0})
+ (\tilde{Y}_i(z_i,\bm{0})-\tilde{Y}_i(0,\bm{0}))
+ (\tilde{Y}_i(0,\mathbf{z}_{\mathcal{N}_i})-\tilde{Y}_i(0,\bm{0}))
\end{equation}
for all treatments $\mathbf{z}$ and all units $i$.
\end{definition}
The first term in Eq.~\eqref{eq:add_main} is the baseline outcome under no treatment, the second term is the direct treatment effect, the usual object of study in the causal inference literature, and the third term is the interference effect.

\subsubsection{Additivity of Interference Effects}
The second assumption is that the interference effects are additive among the interfering units.
\begin{definition}[Additivity of Interference Effects]\label{def:aie}
The potential outcomes satisfy additivity of interference effects if
\begin{align}
 \quad    \tilde{Y}(z_i,\mathbf{z}_{\mathcal{N}_i}) &= \tilde{Y}(z_i,\bm{0}) + \sum_{j=1}^{d_i} \left(\tilde{Y}_i\left(z_i, (\mathbf{z}_{\mathcal{N}_i})_j \mathbf{e}_j\right)-\tilde{Y}\left(z_i,\bm{0}\right)\right),
\end{align}
for all allocations $\mathbf{z}$ and units $i$.
Note that $\mathbf{e}_j\in \{0,1\}^d_i$ is zero in each coordinate except coordinate $j$.
\end{definition}
The first term above is the outcome if no neighbors are treated and each summand is the additional effect of each neighbor's treatment.
As an example, this assumption can be verified to hold in the models proposed in \citet{toulis2013estimation} and \citet{eckles2016estimating} uses closely related assumptions.

\subsubsection{Symmetrically Received Interference Effects}
The third assumption is that the impact of interferences is invariant to which particular subset of neighbors are treated.
\begin{definition}[Symmetrically Received Interference Effects]\label{def:sri}
The potential outcomes have symmetrically received interference if
\begin{equation}
  \tilde{Y}_i(z_i, \mathbf{z}_{\mathcal{N}_i}) = \tilde{Y}_i(z_i, \tau(\mathbf{z}_{\mathcal{N}_i})),
\end{equation}
for all allocations $\mathbf{z}$, units $i$, and permutations $\tau$ of vectors of length $d_i$.
\end{definition}
Equivalently, this means that $\tilde{Y}$ depends on its second argument only through the number of treated neighbors $d_i^{\mathbf{z}}$.
This is also referred to as anonymous interference \citep{manski2013identification} or no peer-effect-heterogeneity \citep{athey2015exact}.

\subsubsection{Symmetrically Sent Inteference Effects}
While the symmetrically received interference effects assumption asserts that the interference effects for unit $i$ are equal across permutations of the neighbors treatments, the symmetrically sent interference assumption asserts that units that share a neighbor are effected ``equally'' by that neighbor.
In general, it is unclear what should be assumed equal since there may be interaction effects between sets of neighbors.
Under additivity of interference effects, however, we can assume that the equality is in terms of the differences between the relevant potential outcomes.
In this manuscript, we will only assume symmetrically sent interference effects in conjunction with additivity of interference effects.
\begin{definition}[Symmetrically Sent Interference Effects]\label{def:ssi}
If the potential outcomes satisfy additivitivity of interference effects (see Definition~\ref{def:aie}), then the potential outcomes have symmetrically sent interference if
\begin{equation}
   Y_i(\mathbf{z}+\mathbf{e}_j)-Y_i(\mathbf{z}) = Y_{i'}(\mathbf{z}+\mathbf{e}_j)-Y_{i'}(\mathbf{z}), \label{eq:ssie}
\end{equation}
for units $j$ and allocations $\mathbf{z}$ where $z_j=0$, and units $i,i'$ where $g_{ji}=g_{ji'}=1$,
where $\mathbf{e}_j$ is the allocation where only unit $j$ is in the treatment group.
\end{definition}
Again, this assumption can be verified to hold in certain models proposed in \citet{toulis2013estimation} and \citet{eckles2016estimating}.

\begin{remark}
In this manuscript, we avoid assumptions of constant treatment effects across units so that our methods may be applicable to realistic settings. 
However, as we will see later, certain combinations of assumptions, such as combining symmetrically sent interference effects and symmetrically received interference effects will imply constant treatment effects across certain units. 
Conversely, if we assume symmetrically sent interference effects and we assume that the interference effects are constant across units then this will also imply that the potential outcomes will satisfy the symmetrically received interference assumption as well.
This arises because the assumption of constant interference effects will impose that the potential outcomes are linear in the treated degree.
Section~\ref{sec:sanasia_param} and Appendix~\ref{app:param} offer more details.\qed
\end{remark}



\section{Models and Estimands} \label{sec:model_estimand}
In practice, an analyst can choose which additional assumptions may hold in addition to NIA.
Figure~\ref{fig:assumGram} shows the partial ordering of the twelve possible subsets of assumptions that can be chosen from the four structural assumptions.
Certain combinations of assumptions  are disallowed because we require that Symmetrically Sent Interference Effects is only assumed in combination with Additivity of Interference Effects.
Each combination of assumptions leads to a particular model or parameterization for the potential outcomes. 
In Section~\ref{sec:param}, we analyze three of these models, NIA with no additional assumptions, SANIA, where we assume symmetrically received interference and additivity of main effects, and SANASIA where we include all four assumptions from Section~\ref{sec:structural}.
We then discuss how these parameterizations can be used to define estimands and how these estimands relate to estimands defined in terms of potential outcomes in Section~\ref{sec:estimand}.
Appendix~\ref{app:param} describes the parameterizations implied by the remaining nine combinations of assumptions.


\begin{figure}
\begin{minipage}[c]{0.6\textwidth}
\newcommand{\uc}{\char`_}

\begin{tikzpicture}[->,>=stealth',shorten >=1pt,auto,node distance=1.75cm,
  thick, column sep=1cm,main node/.style={font=\ttfamily,rectangle, draw=black, fill=none}]


  \node[main node] (NIA) {\char`_\char`_N\char`_\char`_IA};
  \node[main node] (ANIA) [below of=NIA] {\char`_AN\char`_\char`_IA};
  \node[main node] (NAIA) [right of=ANIA] {\char`_\char`_NA\char`_IA};
  \node[main node] (SNIA) [left of=ANIA] {S\char`_N\char`_\char`_IA};
  \node[main node] (ANAIA)[below of=NAIA]{\char`_ANA\char`_IA};
  \node[main node] (SANIA)[below of=SNIA]{SAN\char`_\char`_IA};
  \node[main node] (SNAIA)[below of=ANIA]{S\char`_NA\char`_IA};
  \node[main node] (SANAIA)[below of=SNAIA]{SANA\char`_IA};
  \node[main node] (NASIA)[right of=ANAIA]{\char`_\char`_NASIA};
  \node[main node] (SNASIA)[below of=ANAIA]{S\char`_NASIA};
  \node[main node] (ANASIA)[right of=SNASIA]{\char`_ANASIA};
  \node[main node] (SANASIA)[below of=SNASIA]{SANASIA};

  \path[every node/.style={font=\sffamily\small}]
    (NIA) edge node [left]{} (ANIA)
          edge node [below]{} (NAIA)
          edge node [right]{} (SNIA)
    (ANIA) edge node [below]{} (ANAIA)
          edge node [right]{} (SANIA)
    (NAIA) edge node [left]{} (ANAIA)
          edge node [right]{} (SNAIA)
    (SNIA) edge node [left]{} (SANIA)
          edge node [below]{} (SNAIA)
    (ANAIA) edge node [left]{} (SANAIA)
    (SANIA) edge node [below]{} (SANAIA)
 
    (SNAIA) edge [ right] node [right]{} (SNASIA)
    (NASIA) edge node [right]{} (SNASIA)
    (ANAIA) edge [ right] node [right]{} (ANASIA)
    (SANAIA) edge node [right]{} (SANASIA)
    (SNAIA) edge node [right]{} (SANAIA)
    (NAIA) edge [ right] node[left]{} (NASIA)
    (NASIA) edge node [below]{} (ANASIA)
    (ANASIA) edge node [below]{} (SANASIA)
    (SNASIA) edge node [left]{} (SANASIA);
\end{tikzpicture}
  \end{minipage}\hfill
  \begin{minipage}[c]{0.4\textwidth}
    \caption{
       The partial ordering of models that arise from combining structural assumptions.
       An arrow indicates that one model contains the other.
 Each element is labeled where each of four possibly blank characters refers to whether a particular assumption is made: 
(1) Symmetrically Received Interference Effects, 
(2) Additivity of Main Effects, 
(3) Additivity of Interference Effects, 
(4) Symmetrically Sent Interference Effects.
The largest model is NIA while the smallest model is SANASIA, where all four structural assumptions are made.   } \label{fig:assumGram}
  \end{minipage}
\end{figure}
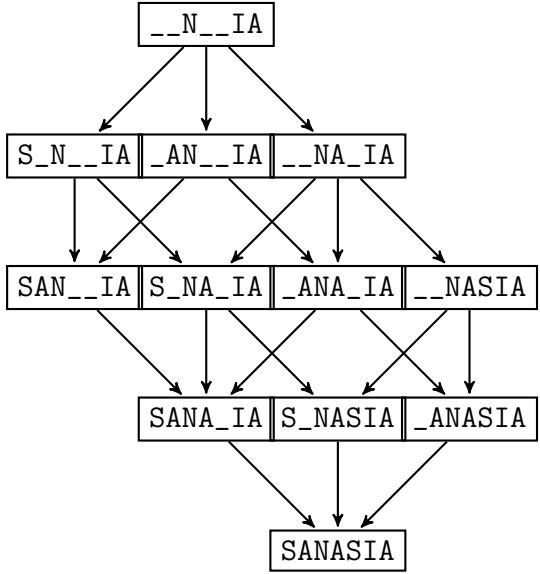

\subsection{Parameterizations for Observed Outcomes}\label{sec:param}

For each of these models, we can parametrize the potential outcomes in terms of the baseline outcome, the direct treatment effect, the interference effect, and possible interaction between direct treatment and interference effects.
We note that the direct treatment effect and the interference effect are closely related to the direct and indirect effects, respeectively, as defined in \citep{hudgens2012toward}, while the interaction effect can be viewed as the difference between the total causal effect and indirect effects as definied there. 

\subsubsection{NIA}

Under NIA, we can write the potential outcome $Y_i(\mathbf{z})$ as 
\begin{align*}
 \overbrace{\tilde{Y}(0,\bm{0})}^{\text{baseline}}&+ z_i \overbrace{(\tilde{Y}(1,\bm{0})-\tilde{Y}(0,\bm{0})}^{\text{direct treatment effect}}) + (\overbrace{\tilde{Y}(0, \mathbf{z}_{\mathcal{N}_i})-\tilde{Y}(0,\bm{0})}^{\text{interference effect}})\\
 & +z_i (\underbrace{\tilde{Y}(1, \mathbf{z}_{\mathcal{N}_i})-\tilde{Y}(1,\bm{0})-(\tilde{Y}(0, \mathbf{z}_{\mathcal{N}_i})-\tilde{Y}(0,\bm{0}))}_{\text{interaction effect}}).
\end{align*}

We define $\alpha_i=\tilde{Y}_i(0,\bm{0})$ as the baseline outcome and $\beta_i=\tilde{Y}_i(1,\bm{0})-\tilde{Y}_i(0,\bm{0})$ as the direct treatment effect.
The functions $\Gamma_i:\{0,1\}^{d_i}\mapsto \Re$ and $\Delta_i:\{0,1\}^{d_i}\mapsto \Re$ are defined as 
\begin{align*}
  \Gamma_i(\mathbf{z}_{\mathcal{N}_i}) &= \tilde{Y}_i(0,\mathbf{z}_{\mathcal{N}_i})-\tilde{Y}_i(0,\bm{0}) \text{ and } \\
  \Delta_i(\mathbf{z}_{\mathcal{N}_i}) &= \tilde{Y}_i(1, \mathbf{z}_{\mathcal{N}_i})-\tilde{Y}_i(1,\bm{0})-(\tilde{Y}_i(0, \mathbf{z}_{\mathcal{N}_i})-\tilde{Y}_i(0,\bm{0})).
\end{align*}
Hence, $\Gamma_i(\mathbf{z}_{\mathcal{N}_i})$ gives the interference effect if the treatment allocation for unit $i$'s neighbors is $\mathbf{z}_{\mathcal{N}_i}$.
Similarly $\Delta_i$ gives the additionally interaction effect between the neighbors' treatments and direct treatment of the unit. 

Using this notation, the NIA assumption is equivalent to the assertion that the potential outcomes adhere to the following parameterization,
\begin{equation}
  Y_i(\mathbf{z}) = \alpha_i+ \beta_i z_i +\Gamma_i(\mathbf{z}_{\mathcal{N}_i})+ z_i \Delta_i(\mathbf{z}_{\mathcal{N}_i}).
\end{equation}
For unit $i$, this parameterization has $2^{d_i+1}$ parameters (noting that $\Gamma_i(\bm{0})=\Delta_i(\bm{0})=0$), which is exactly the number of possibly distinct potential outcomes.

\subsubsection{SANIA}

Under SANIA, where symmetrically received interference effects and additivity of main effects are assumed, we abuse notation and redefine $\Gamma_i:\{0,1,\dotsc,d_i\} \mapsto \Re$ such that $\Gamma_i(d) = \tilde{Y}_i(0,\mathbf{z}_{\mathcal{N}_i})-\tilde{Y}_i(0,\bm{0})$ for all $\mathbf{z}$ where $\sum_{j\in \mathcal{N}_i} z_i =d$.
This is well defined due to the symmetrically received interference effects assumption.
Since the main effects are additive, we may also omit the $\Delta$ term, yielding
\begin{equation}
  Y_i(\mathbf{z}) = \alpha_i+ \beta_i z_i + \Gamma_i(d_i^{\mathbf{z}}) \label{eq:param_sania}
\end{equation}
for all allocations $\mathbf{z}$.
Here we have $2+d_i$ parameters which encode $2d_i+2$ distinct potential outcomes for unit $i$.
Our simulation in Section~\ref{sec:simulation} and much of our theoretical results are focused on this particular model which is relatively parsimonious.

\subsubsection{SANASIA}\label{sec:sanasia_param}
The SANASIA assumption restricts SANIA further by assuming additivity of interference effects and symmetrically sent interference effect. 
This means that the functions $\Gamma_i$ must be additive.
Since we restrict that $\Gamma_i(0)=0$, we have that $\Gamma_i$ is linear so that $\Gamma_i(d)=\gamma_i d$ for some $\gamma_i\in\Re$.
Hence, the potential outcomes must satisfy
\begin{equation}
    Y_i(\mathbf{z}) = \alpha_i + \beta_i z_i +\gamma_i d_i^{\mathbf{z}}. \label{eq:paramSANASIA}
\end{equation}
Additionally, the interference effects are symmetrically sent.
This implies that if $i,i'$ both have $j$ as a neighbor then the additive effect of $j$ being treated is the same for both $i$ and $i'$, so that $\gamma_i=\gamma_{i'}$.
By iterating on this idea, one can show the following proposition which relies on construction of the undirected graph $h$ which indicates whether two nodes share a neighbor.

\begin{proposition}\label{prop:paramSANASIA}
Let $h\in\{0,1\}^{n\times n}$ be an adjacency matrix where $h_{ij}=\bm{1}\{(g^T g)_{ij}>0\}$.
Under SANASIA, the potential outcomes can be parametrized as 
$Y_i(\mathbf{z})= \alpha_i+\beta_i z_i+\gamma_i d_i^{\mathbf{z}}$
where if $i$ and $i'$ are in the same connected component of $h$ then $\gamma_i=\gamma_{i'}$.

\end{proposition}
Since the interference parameters are equal in connected components of $h$, the total number of parameters needed for all potential outcomes is $2n$ plus the number of connected components of $h$.
In the case that $g$ is undirected and connected, if $g$ is bipartite then $h$ has two connected components, otherwise $h$ has one connected component.
In order to parametrize the potential outcomes in the latter case, we need only $2n +1$ parameters with $Y_i(\mathbf{z})=\alpha_i+\beta_i z_i + \gamma d_i^{\textbf{z}}$ for each each unit.
These parameters encode $2 n+2\sum_i d_i$ distinct potential outcomes.

\subsection{Estimands}\label{sec:estimand}
Causal estimands are defined as functions of the potential outcomes, frequently an average of a given function over a set of units.  
Given parametric forms for the potential outcomes, we can define estimands in terms these parameters as well: (i) estimands for direct treatment effects, the $\beta_i$ terms, (ii) estimands for interference effects, the $\Gamma_i$ terms, and (iii) estimands for interaction effects, the $\Delta_i$ terms. 
Estimands defined in terms of potential outcomes map to different functions of the parameters depending on which assumptions hold, which we demonstrate below for the three parameterizations from Section~\ref{sec:param}.

Combinations of these effects, such as the mean effect of treating all units versus treating no units can be identified via potential outcomes as
\[
  \frac{1}{n} \sum_{i=1}^n Y_i(\bm{1})-Y_i(\bm{0}).
\]
Under NIA this estimand is identified parametrically as $\frac{1}{n}\sum_i \beta_i+ \Gamma_i(\mathbf{1})+\Delta_i(\mathbf{1})$.
On the other hand, if we assume SANIA holds, this estimand is $\frac{1}{n}\sum_i \beta_i +\Gamma_i(d_i)$, while under SANASIA, this estimand is $\frac{1}{n}\sum_i \beta_i+ \gamma_i d_i$.
Finally, if SUTVA holds then the above estimand is simply $\bar{\beta}=\frac{1}{n}\sum_i \beta_i$.
Alternatively, the same estimand defined in terms of parameters, such as $\bar{\beta}$, will potentially have multiple definitions in terms of potential outcomes depending on the specific assumptions made.
Defining estimands from potential outcomes or from parameterizations are complementary approaches that are differentially impacted by changes in the underlying assumptions.

For the remainder of this manuscript, we will focus on estimation of $\bar{\beta}$, though the core ideas presented below apply to other estimands.
The estimand $\bar{\beta}$ can be identified as the average direct treatment effect and under NIA and all submodels, including SUTVA, $\bar{\beta}$ can be defined in terms of the potential outcomes as
\begin{equation}
  \bar{\beta} = \frac{1}{n} \sum_{i=1}^n Y_i(\mathbf{e}_i)- Y_i(\mathbf{0})
\end{equation}
where under $\mathbf{e}_i$ only unit $i$ receives treatment.
If we know SUTVA holds, then \[
  \bar{\beta}=\frac{1}{n} \sum_{i=1}^n Y_i(1)-Y_i(0) = \frac{1}{n} \sum_{i=1}^n Y_i(\mathbf{1})-Y_i(\mathbf{0})=\frac{1}{n} \sum_{i=1}^n Y_i(\mathbf{e}_i)- Y_i(\mathbf{0}).
\]
For models between NIA and SUTVA, we can define $\bar{\beta}$ in terms of other potential outcomes, possibly depending on the structure of the graph.




\section{Linear Unbiased Estimation}\label{sec:unbiased}
In this section we will answer the question of when unbiased estimators of the average direct treatment effect exist. 
We also focus on linear estimators which are of the form $\hat{\beta}^{\mathbf{w}} = \sum_{i=1}^n w_i(\mathbf{z}) Y_i(\mathbf{z})$ where each $w_i:\mathcal{Z}^n \mapsto \Re$ is a function that gives the coefficient for unit $i$ under allocation $\mathbf{z}$.
Non-linear estimators can also be unbiased but since the estimands are linear functions of the potential outcomes, linear estimators are a natural starting point for study.
Linear estimators include the naive difference of means estimator where 
\begin{equation}
     w_i(\mathbf{z}) = \frac{z_i}{\sum_j z_j} - \frac{1-z_i}{\sum_j (1-z_j)}, \label{eq:est_naive}
 \end{equation}
 and the Horvitz-Thompson inverse propensity score weighting estimator where 
 \begin{equation}
      w_i(\mathbf{z}) = \frac{z_i}{n \Pr[z_i^{obs}=z_i]} - \frac{1-z_i}{n\Pr[z_i^{obs}=z_i]}, \label{eq:est_ht}
  \end{equation}
as well as stratified estimators (see Eq.~\ref{eq:strat_naive} in Proposition~\ref{prop:vertex_trans}) and other more sophisticated estimators.

Even under SUTVA, certain pathological designs will preclude the existence of unbiased estimators. 
In Section~\ref{sec:ue_dte}, we'll give conditions on the design for when linear unbiased estimates exists under NIA, SANIA, and SANASIA.
In Section~\ref{sec:designExamples}, we will analyze completely randomized designs and Bernoulli trial as well as various estimators in the context of NIA.



\subsection{Unbiased Estimators of Direct Treatment Effect}\label{sec:ue_dte}

We will require that  $\hat{\beta}^{\mathbf{w}}$ be unbiased for $\bar{\beta}$ with respect to the known design.
For linear estimators, the requirement that $\hat{\beta}^{\mathbf{w}}$ is unbiased is equivalent to a set of linear constraints on the coefficients.
The particular constraints depend on which assumptions we make.

\subsubsection{NIA}
\begin{proposition}\label{prop:constraintNIA}
Under NIA, a linear estimator $\hat{\beta}^{\mathbf{w}}$ is unbiased for $\bar{\beta}$ if and only if 
for all $i\in[n]$, the following constraints hold:
\begin{align*}
&\quad \sum_{\mathbf{z}\in \mathcal{Z}} p(\mathbf{z}) w_i(\mathbf{z}) z_i = 1/n.	\tag{$\beta_i$ constraints) (C1} \\
&\quad \sum_{\bf{z}\in \mathcal{Z}} p(\mathbf{z}) w_i(\mathbf{z} ) = 0,	\tag{$\alpha_i$ constraints) (C2}\\
\forall \mathbf{z}'\in \{0,1\}^{d_i}\setminus\{0\},& \quad \sum_{\bf{z}\in \mathcal{Z}} p(\mathbf{z})  w_{i}(\mathbf{z}) \bm{1}\{\mathbf{z}_{\mathcal{N}_i}=\mathbf{z}'\}  = 0, \tag{$\Gamma_i(\mathbf{z}')$ constraints) (C3}\\
\text{and} \quad \forall \mathbf{z}'\in \{0,1\}^{d_i}\setminus \{0\},& \quad \sum_{\bf{z}\in \mathcal{Z}} p(\mathbf{z})  w_{i}(\mathbf{z})z_i \bm{1}\{\mathbf{z}_{\mathcal{N}_i}=\mathbf{z}'\}  = 0, 	\tag{$\Delta_i(\mathbf{z}')$ constraints) (C4}
\end{align*}
\end{proposition}
The first line, C1, ensures that the weights for $\beta_i$ terms average to $1/n$ across allocations which guarantees the weights in front all $\beta$ terms average to unity across allocations.
The last three lines, C2, C3, and C4, ensure all other parameters will have weights that average to zero. 
Due to the fact that we allow for direct treatment effect heterogeneity, these constraints are also equivalent to the constraint that for each $i\in[n]$, the estimator $\hat{\beta}_i = w_i(\mathbf{z}^{obs}) Y_i(\mathbf{z}^{obs})$ is unbiased for $\beta_i/n$.

By manipulating the constraints above, we can find a concise condition on the design $p$ such that unbiased estimators of $\bar{\beta}$ exist.
First, by combining the C4 constraints and subtracting them from the C1 constraints we have that for all $i\in [n]$ it must hold that
\[
   \sum_{\bf{z}\in \mathcal{Z}} p(\mathbf{z})  w_{i}(\mathbf{z})z_i \bm{1}\{ d_i^{\mathbf{z}}=0\}  = 1/n \text{ and} \sum_{\bf{z}\in \mathcal{Z}} p(\mathbf{z})  w_{i}(\mathbf{z})(1-z_i) \bm{1}\{ d_i^{\mathbf{z}}=0\}  = -1/n
\]
This leads to the following results.

\begin{proposition}\label{prop:existNIA}
Under NIA, LUEs exist if and only if for each $i\in[n]$ there exist allocations $\mathbf{z},\mathbf{z}'$ such that $p(\mathbf{z})>0$, $p(\mathbf{z}')>0$, $d_i^\mathbf{z}=d_i^{\mathbf{z}'}=0$, $z_i=1$, and $z_i'=0$.
\end{proposition}
\begin{proof}
For the necessity, if no such allocations exist, then one of the summations in the displayed equation above would necessarily sum to $0$ rather than $\pm 1/n$. 
The sufficiency follows by considering inverse propensity score weighting. 
Namely, for the two allocations that satisfy the conditions, set $w_i(\mathbf{z})=p(\mathbf{z})^{-1}/n$, $w_i(\mathbf{z}')=-p(\mathbf{z})^{-1}/n$, and $w_i(\mathbf{z}'')=0$ for all other allocations.
\end{proof}

This propositions implies that, regardless of the network structure, one can construct a design such that unbiased estimators exist.
By ensuring the support of the design contains $\{\mathbf{e}_i: i\in[n]\} \cup \{\mathbf{0}\}$ where $\mathbf{e}_i$ is the allocation where only unit $i$ is treated, the constraints of Proposition~\ref{prop:existNIA} are satisfied.
In particular, unbiased estimators will always exist for non-trivial Bernoulli trials.

\begin{remark}
More generally we can consider designs derived from a collection of independent sets, which are subsets of vertices where no two vertices in the same subset are adjacent.
An independent set for a network $g$ is a set $\mathcal{V}\subset [n]$ such that for all $i,j\in \mathcal{V}$, $g_{ij}=0$, meaning  there are no edges between units in $\mathcal{V}$.
Let $\mathcal{V}_1,\dotsc,\mathcal{V}_M$ be a collection of disjoint independent sets whose union is $[n]$.
Such a collection is called a proper vertex coloring and no two vertices of the same color are adjacent.
An example of a proper vertex coloring is demonstrated in Figure~\ref{fig:ind_set}.

Proper vertex colorings exist for every graph because the singleton $\{i\}$ is an independent set and hence $\{1\},\{2\},\dotsc, \{n\}$ is a proper vertex coloring.
One can verify that the conditions of Proposition~\ref{prop:existNIA} will be satisfied by a design supported on $\bm{0}, \mathbf{z}^{(1)},\dotsc, \mathbf{z}^{(M)}$, where for each $m$, $\mathbf{z}^{(m)}$ is the allocation with treatment group $\mathcal{V}_m$.
Designs based off of independent sets ensure unbiased estimates exist and they can also be used to minimize the presence of interference effects.\qed
\end{remark}

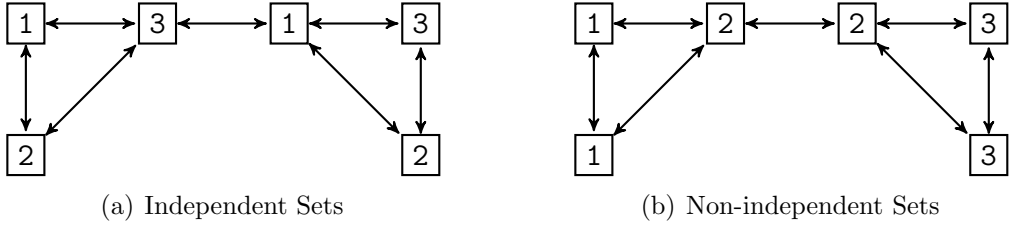
\begin{figure*}[tb]
	\hfill
	\subfigure[Independent Sets]{
	\begin{tikzpicture}[<->,>=stealth',shorten >=1pt,auto,node distance=1.75cm,
	  thick, column sep=1cm,main node/.style={font=\ttfamily,rectangle, draw=black, fill=none}]
	
	  \node[main node] (one) {1};
	  \node[main node] (two) [below of=one] {2};
	  \node[main node] (three) [right of=one] {3};
	  \node[main node] (four) [right of=three] {1};
	  \node[main node] (five) [right of=four] {3};
	  \node[main node] (six) [below of=five] {2};
	  
	  \path[every node/.style={font=\sffamily\small}]
	    (one) edge node [left]{} (two)
	          edge node [below]{} (three)
	    (two) edge node [below]{} (three)
	    (three) edge node [left]{} (four)
	    (four) edge node [left]{} (five)
	    (four) edge node [left]{} (six)
	    (six) edge node [left]{} (five);
	\end{tikzpicture}
	}
\hfill
	\subfigure[Non-independent Sets]{
	\begin{tikzpicture}[<->,>=stealth',shorten >=1pt,auto,node distance=1.75cm,
	  thick, column sep=1cm,main node/.style={font=\ttfamily,rectangle, draw=black, fill=none}]
	
	  \node[main node] (one) {1};
	  \node[main node] (two) [below of=one] {1};
	  \node[main node] (three) [right of=one] {2};
	  \node[main node] (four) [right of=three] {2};
	  \node[main node] (five) [right of=four] {3};
	  \node[main node] (six) [below of=five] {3};

	  \path[every node/.style={font=\sffamily\small}]
	    (one) edge node [left]{} (two)
	          edge node [below]{} (three)
	    (two) edge node [below]{} (three)
	    (three) edge node [left]{} (four)
	    (four) edge node [left]{} (five)
	    (four) edge node [left]{} (six)
	    (six) edge node [left]{} (five);
	\end{tikzpicture}
	}	\hfill ~\!

	\caption{(a) The left panel shows three independent sets, labeled 1, 2, and 3, which partition the graph and form a proper vertex coloring. Notice that no vertex is adjacent to a vertex with the same number. (b) The right panel shows a partition but the sets are not independent sets since vertices are adjacent to other vertices with the same partition.}
	\label{fig:ind_set}
\end{figure*}

In Section~\ref{sec:designExamples} we will examine whether the constraints from Proposition~\ref{prop:existNIA} are satisfied for standard designs for certain networks.
Next, we investigate the existence of unbiased estimators under SANIA, but note that merely assuming symmetrically received interference, the conditions for the existence of  unbiased estimators are the same as in Proposition~\ref{prop:existNIA}.

\subsubsection{SANIA}
Under SANIA, where $Y_i(\mathbf{z}) = \alpha_i+\beta_i+\Gamma_i(d_i^{\mathbf{z}})$, the constraints C1 and C2  from Proposition~\ref{prop:constraintNIA} remain unchanged and constraint C4 is vacuous.
The constraint C3 for the $\Gamma_i$ coefficients are replaced with for all $i\in [n]$ it holds that 
\begin{align*}
\forall d\in [d_i],& \quad \sum_{\bf{z}\in \mathcal{Z}} p(\mathbf{z})  w_{i}(\mathbf{z}) \bm{1}\{ d_i^\mathbf{z}=d\}  = 0. \tag{$\Gamma_i(d)$ coefficients) (C3$'$}
\end{align*}
The exclusion of interaction effects due to the additivity of main effects, Definition~\ref{def:add}, as well as the symmetrically received interference effects, Definition~\ref{def:sri}, affords much more flexibility for unbiased estimates yielding weaker constraints for their existence.
\begin{proposition}\label{prop:existSANIA}
Under SANIA,  LUEs exist if and only if for each $i\in[n]$ there exist allocations $\mathbf{z},\mathbf{z}'$ such that $p(\mathbf{z})>0$, $p(\mathbf{z}')>0$, $d_i^\mathbf{z}=d_i^{\mathbf{z}'}$, $z_i=1$, and $z_i'=0$.
\end{proposition}
While designs using independent set designs will allow for the existence of unbiased estimated under SANIA, much more flexible designs will  yield unbiased estimators and tractable estimation procedures, as will be seen in the Section~\ref{sec:mivlue}.

\subsubsection{SANASIA} 
\label{ssub:sanasiaUE}
The constraint C1 and C2 on the $\alpha_i$ and $\beta_i$ terms, respectively, are the same for SANASIA as under NIA and SANIA.
Let $\mathcal{C}_1,\dotsc, \mathcal{C}_K$ be the connected components of $h=\bm{I}\{ g^T g>0\}$, as described in Proposition~\ref{prop:paramSANASIA}.
In this case, the only other assumption to ensure that $\hat{\beta}^{\mathbf{w}}$ is unbiased for $\beta$ is that 
\[
 \forall k\in[K],\quad  \sum_{\mathbf{z}} p(\mathbf{z}) \sum_{i \in \mathcal{C}_k} w_i(\mathbf{z})d_i^{\mathbf{z}} = 0. \tag{$\gamma_i$ constraint) (C3$''$}
\]

Due to the nested nature of SANASIA and SANIA, a sufficient condition for unbiased estimates to exist is that the design satisfies Proposition~\ref{prop:existSANIA}.
A necessary but insufficient condition is that for each unit $i$ there exist allocations $\mathbf{z},\mathbf{z}'$ such that $p(\mathbf{z})>0$, $p(\mathbf{z}')>0$ and $z_i=0$ and $z'_i=1$, the condition for the existence of unibased estimates under SUTVA. 
Under SANASIA, we have not found a concise necessary and sufficient condition for the existence unbiased estimators given a design.

\subsection{Examples of Biased and Unbiased Estimates}\label{sec:designExamples}
In this section we will examine Bernoulli trials, completely randomized designs and other designs as well as standard estimators such as the naive estimator and Horvitz-Thompson estimator in the context of small examples.
\begin{example}[Empty Graph, SUTVA]
First, consider the case where the graph is empty, so that $g_{ij}=0$ for all $i,j$, or equivalently that SUTVA holds.
In this case unbiased estimators exist as long as for each unit $i$ there exists allocations $\mathbf{z},\mathbf{z}'$ such that $z_i=1=1-z'_i$ and $p(\mathbf{z}), p(\mathbf{z}')>0$. 
Hence, unbiased estimators exists for any nontrivial Bernoulli trial and completely randomized design.
A standard results is that the Horvitz-Thompson, Eq.~\eqref{eq:est_ht}, is always unbiased if unbiased estimators exist. 
The naive estimator Eq.~\eqref{eq:est_naive}, will be unbiased for any design which is invariant under permutations of the unit labels, which can be proved by an application of Proposition~\ref{prop:vertex_trans} to the case of the empty graph.
\qed
\end{example}

\begin{example}[Complete Graph]
 On the other end of the spectrum, suppose that the graph is complete so that $g_{ij}=1$ for all $i\neq j$. 
 In this case, NIA does not constrain the interference and by Proposition~\ref{prop:existNIA} unbiased estimators exist if and only if $p(\mathbf{e}_i)>0$ for all $i$ and $p(\bm{0})>0$. 
 Note that a Bernoulli trial will satisfy this but a completely randomized design will not because $p(\mathbf{0})=0$ for all but trivial completely randomized designs.
 For any design where unbiased estimators exist, the Horvitz-Thompson estimator with \[ w_i(\mathbf{z}) = p(\mathbf{e}_i)^{-1} \mathbf{I}\{\mathbf{z}=\mathbf{e}_i\} - p(\bm{0})^{-1}\mathbf{I}\{\mathbf{z}=\bm{0}\}, \] will be unbiased. 
 On the other hand, the naive difference of means estimator will never be unbiased in this case.

If SANIA holds, completely randomized designs also do not yield unbiased estimators because the treated degree of a unit is determined by its treatment.
For a completely randomize design with $k$ treated units, a unit that receives treatment will have $k-1$ treated neighbors whereas a unit that receives control will have $k$ treated neighbors, and hence the conditions of Proposition~\ref{prop:existSANIA} will not hold.
It is also straightforward to show that the same is true even under SANASIA. 
If we consider a mixture of two completely randomized designed, one where $n_t$ units are treated and one where $n_t+1$ units are treated, then in this case unbiased estimators will exist for any $n_t < n$.  \qed
 \end{example} 

As shown in \citet{karwa2016bias}, even under SANASIA with the added assumption of constant treatment effects, the naive estimator will be biased for any Bernoulli trial or completely randomized design for any non-empty graph. 
On the other hand, while the standard Horvitz-Thompson estimator under SUTVA may be biased, the use of inverse propensity score weighting will be unbiased in some generality \citep{aronow2013estimating}.

\begin{example}[Triangle plus Tail]\label{ex:tri_w_tail_exist}

Consider an undirected network on four nodes where the first three nodes are all connected and the fourth node is adjacent only to the third node, as illustrated in Figure~\ref{fig:tri_w_tail}. 
For this network, under SANIA, again the a Bernoulli trial will yield unbiased estimates while unbiased estimates do not exist for any completely randomized design because as in the complete graph example, the third node has all other nodes as neighbors. 
Table~\ref{tab:tri_w_tail} shows the values of the potential outcomes under SANIA for all units and all allocations in a completely randomized design with two units treated.
Hence, if the support of the design distribution is any subset of the allocations with exactly two treated units, then no unbiased estimator will exist since the treated degree for unit three is always two if the unit is in the control group and the treated degree is one if the unit is in the treatment group.
This violates Proposition~\ref{prop:existSANIA}. \qed

\begin{figure}
\label{fig:tri_w_tail}
\begin{center}
\begin{tikzpicture}[<->,>=stealth',shorten >=1pt,auto,node distance=1.75cm,
  thick, column sep=1cm,main node/.style={font=\ttfamily,rectangle, draw=black, fill=none}]

  \node[main node] (one) {1};
  \node[main node] (two) [below of=one] {2};
  \node[main node] (three) [right of=one] {3};
  \node[main node] (four) [right of=three] {4};

  \path[every node/.style={font=\sffamily\small}]
    (one) edge node [left]{} (two)
          edge node [below]{} (three)
    (two) edge node [below]{} (three)
    (three) edge node [left]{} (four);
\end{tikzpicture}
\end{center}
\caption{Triangle plus tail graph as described in Example~\ref{ex:tri_w_tail_exist}.}
\end{figure}
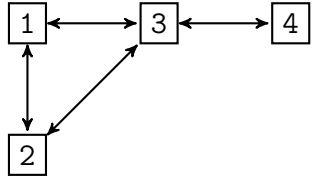

\begin{table}
\begin{center}
\small
\begin{tabular}{c|rrrr} \toprule 
$\mathbf{z}$ & $Y_1^{obs}$ & $Y_2^{obs}$ & $Y_3^{obs}$ & $Y_4^{obs}$ \\ \midrule
$(1,1,0,0)$ & $\beta_1+\Gamma_1(1)+\alpha_1$ & $\beta_2+\Gamma_2(1)+\alpha_2$ & $\Gamma_3(2)+\alpha_3$ & $\alpha_4$ \\
$(1,0,1,0)$ & $\beta_1+\Gamma_1(1)+\alpha_1$ & $\Gamma_2(2)+\alpha_2$ & $\beta_3+\Gamma_3(1)+\alpha_3$ & $\Gamma_4(1)+\alpha_4$ \\ 
$(1,0,0,1)$ & $\beta_1+\alpha_1$ & $\Gamma_2(1)+\alpha_2$ & $\Gamma_3(2)+\alpha_3$ & $\beta_4+\alpha_4$ \\ 
$(0,1,1,0)$ & $\Gamma_1(2)+\alpha_1$ & $\beta_2+\Gamma_2(1)+\alpha_2$ & $\beta_3+\Gamma_3(1)+\alpha_3$ & $\Gamma_4(1)+\alpha_4$ \\ 
$(0,1,0,1)$ & $\Gamma_1(1)+\alpha_1$ & $\beta_2+\alpha_2$ & $\Gamma_3(2)+\alpha_3$ & $\beta_4+\alpha_4$  \\ 
$(0,0,1,1)$ & $\Gamma_1(1)+\alpha_1$ & $\Gamma_2(1)+\alpha_2$ & $\beta_3+\Gamma_3(1)+\alpha_3$ & $\beta_4+\Gamma_4(1)+\alpha_4$\\ \bottomrule
\end{tabular}
\end{center}
\caption{The potential outcomes for the network in Example~\ref{ex:tri_w_tail_exist} for the six treatment assignments that assign exactly two units to treatment under SANIA.
\label{tab:tri_w_tail}}
\end{table}
\end{example}





\section{Minimum Integrated Variance Estimation}\label{sec:mivlue}
At this point we have established the constraints on a linear estimator such that it is unbiased and the constraints on the design such that LUEs exist.
Typically this unbiased estimator will not be unique and so we must introduce additional criteria to chose among them. 
Adopting the variance as our measure of performance for an unbiased estimate, we would ideally minimize variance for all values of the parameters to achieve a uniformly minimum variance LUE. 
However, this project is futile as such estimates do not exist even if one assumes there is no interference.
Hence, we adopt an alternative notion of optimality motivated by Bayesian considerations. 

We propose minimizing the integrated variance (MIV) with respect to distributions on the parameters defined in Section~\ref{sec:param}.
These distributions can be viewed as a prior distribution or as simply a weighting of the parameter space where accuracy is most desired.
Note, one may want to compute the best LUE with respect to a posterior distribution of the parameters that depends on the observed treatments and outcomes, however this will not be guaranteed to be unbiased since the particular unbiased estimator that will be used for each allocation can be different and hence bias can be introduced.\footnote{For each allocation, let $\hat{\beta}^\mathbf{z}$ be the MIV LUE for the posterior when the observed allocation is $\mathbf{z}$. While the constraints will ensure that $\sum_{\mathbf{z}'}p(\mathbf{z}') \hat{\beta}^\mathbf{z}(\mathbf{z}')=\bar{\beta}$ the actual bias of this procedure is $\sum_\mathbf{z} p(\mathbf{z}) \hat{\beta}^{\mathbf{z}}(\mathbf{z}) - \bar{\beta}$ where note that the estimator is different for each allocation $\mathbf{z}$.
Hence, we cannot ensure the overall unbiasedness of the procedure.}
For convenience, we will refer to the distribution on the parameters as the prior distribution.

We will introduce some additional notation used throughout this section.
We denote covariance matrices for the parameters with respect to the prior as 
\begin{align*}
\Sigma_{\xi} &= \mathrm{Cov}(\xi,\xi)\in\Re^{n\times n}\text{ and } \Sigma_{\xi,\xi'} = \mathrm{Cov}(\xi,\xi') \in\Re^{n\times n}
\end{align*}
where $\xi,\xi'$ will denote parameter vectors in $\Re^n$.
For example, under SANIA $\xi,\xi'$ can be any of $\alpha,\beta$, $\Gamma(1),\dotsc, \Gamma(n-1)\in \Re^n$ and $\Sigma_{\alpha}$ denotes the prior covariance for the vector $\alpha=(\alpha_1,\dotsc,\alpha_n)$ and similarly $\Sigma_{\alpha \beta}=\mathrm{Cov}(\alpha,\beta)$.
As is typical for linear estimates \citep{bickel2015mathematical,hoff2009first}, the MIV LUE will depend only on the covariance of the prior and not the explicit form of the prior distribution or higher moments so we do not specify these.


To illustrate the ideas, we'll begin by giving results for MIV LUEs in the SUTVA case before extending these results to SANIA, NIA, and SANASIA.
We'll also focus our attention  on two extreme settings which correspond loosely to maximum and minimum heterogeneity.
The first is priors where parameters are independent across units, where analytical results are straightforward.
The second is priors where with probability one the parameters are equal across units.

\subsection{SUTVA} \label{sub:mivlueSUTVA}

Recall that for SUTVA the parameters are $\alpha=(\alpha_1,\dotsc, \alpha_n)$ and $\beta=(\beta_1,\dotsc,\beta_n)$ and $Y_i(\mathbf{z})=\alpha_i + z_i \beta_i$.
For SUTVA and all other models, it is sufficient to consider mean zero priors on the parameters of the model.
If the priors are not mean zero, then the estimate \[ 
\hat{\beta} = \sum_{i=1}^n w_i(\mathbf{z})\left(Y_i(\mathbf{z}) - \mu_{Y_i(\mathbf{z})}\right) + \mu_{\bar{\beta}}
\]
will be unbiased provided the estimate $\hat{\beta}^{\mathbf{w}}$ is unbiased, where $\mu_{Y_i(\mathbf{z})}$ and $\mu_{\bar{\beta}}$ denote the prior mean of the potential outcome $Y_i(\mathbf{z})$ and average direct treatment effect $\bar{\beta}$, respectively.
Additionally, if $\hat{\beta}^{\mathbf{w}}$ minimizes the integrated variance for a mean zero prior with the same covariance, then the above estimator will minimize variance for the shifted prior. 
This follows from standard arguments from Bayesian estimation \citep{hoff2009first}.

If the priors on $(\alpha_i,\beta_i)$ are mean zero then the integrated variance is 
\begin{equation}
  \mathrm{IVAR}(\hat{\beta}^{\mathbf{w}})=\sum_{\mathbf{z}} p(\mathbf{z}) \sum_{i,j} w_i(\mathbf{z})w_j(\mathbf{z})(\Sigma_{\alpha,ij} +z_i z_j\Sigma_{\beta,ij}+z_j\Sigma_{\alpha\beta,ij} + z_i\Sigma_{\beta\alpha,ij}).
\end{equation}
There are two extreme cases to consider for the priors, one corresponding to a fully heterogeneous outcomes and effects, and one corresponding to constant outcomes and effects.
The first is the case where the potential outcomes are identically distributed and uncorrelated across units so that \[
  \Sigma_{\alpha,ij} = \mathbf{I}\{i=j\} \sigma_{\alpha}^2, \quad 
  \Sigma_{\beta,ij} = \mathbf{I}\{i=j\} \sigma_{\beta}^2, \text{ and } 
  \Sigma_{\alpha\beta,ij}= \mathbf{I}\{i=j\} \sigma_{\alpha,\beta}^2.
\]
The integrated variance simplifies to $\sum_{\mathbf{z}} p(\mathbf{z}) \sum_{i=1}^n w_i(\mathbf{z})^2 (\sigma_\alpha^2 + z_i( \sigma_\beta^2 + \sigma_{\alpha,\beta}^2))$.
For this prior, the Horvitz-Thompson estimator is the MIV LUE, see Theorem~\ref{thm:mivlueSUTVA} part 1.

Alternatively, we can specify a prior where we assume $\alpha_i=\alpha_j$ and $\beta_i=\beta_j$ for all $i,j$.
This is the prior for constant treatment effects with constant baseline outcomes and 
\[
  \Sigma_{\alpha,ij} = \sigma_{\alpha}^2, \Sigma_{\beta,ij} = \sigma_{\beta}^2, \text{ and } \Sigma_{\alpha\beta,ij}= \sigma_{\alpha,\beta}^2.
\]
For nontrivial completely randomized designs and Bernoulli trials excluding the allocations where either all units or no units are treated, the MIV LUE for this prior is the naive estimate, see Theorem~\ref{thm:mivlueSUTVA} part 2.
More generally, this holds for symmetric designs where $p(\mathbf{z})=p(\mathbf{z}')$ whenever $\mathbf{z}$ is a permutation of $\mathbf{z}'$, which are exactly designs which are mixtures of completely randomized designs.

The formal statements of these results follow.

\begin{theorem}\label{thm:mivlueSUTVA}
Suppose that the potential outcomes satisfy SUTVA and the design admits unbiased estimators.
\begin{enumerate}
\item For a prior where the parameters are mean zero and uncorrelated across units, the MIV LUE of $\bar{\beta}$ is the Horvitz-Thompson estimator with 
\[
  w_i(\mathbf{z}) = \frac{z_i}{\sum_{\mathbf{z}'} z_i' p(\mathbf{z}')}- \frac{1-z_i}{\sum_{\mathbf{z}'} (1-z_i') p(\mathbf{z}')} = \frac{2z_i-1}{\Pr[z_i^{obs}=z_i]}.
\]
\item For a prior where the parameters are constant across units, if the design is symmetric and does not contain $\bm{0}$ or $\bm{1}$ in its support, then the MIV LUE for $\bar{\beta}$ is the naive estimator with 
\[
w_i(\mathbf{z}) =  \frac{z_i}{\sum_j z_j} - \frac{1-z_i}{n-\sum_j z_j}.
\]
\end{enumerate}
\end{theorem}
We will not prove this theorem directly, but instead apply results from the SANIA setting where the network is empty.
Part 1 of the above is an application of Theorem~\ref{thm:mivlueSANIA}, and part 2 is an application of Proposition~\ref{prop:vertex_trans}.

\subsection{SANIA}\label{sub:mivSANIA}
We now return to SANIA, where the main effects are additive, Definition~\ref{def:add}, and the interference effects are symmetrically received, Definition~\ref{def:sri}.
The potential outcomes satisfy $Y_i(\mathbf{z})=\alpha_i+\beta_i z_i + \Gamma_i(d_i^{\mathbf{z}})$,s see Eq.~\eqref{eq:param_sania}.
For a given allocation $\mathbf{z}$, the integrated square error for $\hat{\beta}^{\mathbf{w}}$ is $\mathbf{w}^T \Sigma(\mathbf{z}) \mathbf{w}-\frac{1}{n^2} \bm{1}^T \Sigma_\beta \bm{1}$ where 
\begin{equation}\label{eq:sigma_z}
\begin{split}
  \Sigma(\mathbf{z}) =&
  \Sigma_\alpha + \mathbf{z}^T \Sigma_\beta \mathbf{z} + \mathbf{z}^T \Sigma_{\beta, \alpha} + \Sigma_{\alpha,\beta}\mathbf{z} \\
  & + \sum_{d=1}^{n-1} \sum_{d'=1}^n \mathbf{I}\{d^{\mathbf{z}}=d\}^T \Sigma_{\Gamma(d),\Gamma(d')}\mathbf{I}\{d^{\mathbf{z}}=d'\} \\
  &+ \sum_{d=1}^{n-1} \mathbf{I}\{d^{\mathbf{z}}=d\}^T \Sigma_{\Gamma(d),\alpha} + \Sigma_{\alpha,\Gamma(d')}\mathbf{I}\{d^{\mathbf{z}}=d\}\\
  &+ \sum_{d=1}^{n-1} \mathbf{I}\{d^{\mathbf{z}}=d\}^T \Sigma_{\Gamma(d),\beta} \mathbf{z} + \mathbf{z}^T \Sigma_{\beta,\Gamma(d')}\mathbf{I}\{d^{\mathbf{z}}=d\}.
\end{split}
\end{equation}
We will now consider first the special case of uncorrelated priors where an explicit solution can be derived in the SANIA and NIA case.
We then analyze the general case where we can compute the solution, but we cannot derive a closed form except for a special case.

\subsubsection{Uncorrelated Prior} 
\label{ssub:uncorrelated_prior}
If the parameters are uncorrelated across units, but possibly correlated within units, then the integrated variance simplifies to 
\[
   \sum_{\mathbf{z}} p(\mathbf{z}) \sum_{i=1}^n w_i(\mathbf{z})^2 \Sigma(\mathbf{z})_{ii} - \frac{1}{n}\sum_i \Sigma_{\beta,ii}
\]
where we require that $\Sigma(\mathbf{z})_{ii}>0$.
Using the method of Lagrange multipliers, we derive an analytic form for the optimal weights $\mathbf{w}$ as a function of the covariances of the parameters and the design.
As under SUTVA, this will lead to a MIV LUE which uses inverse propensity score weighting.

\begin{theorem}\label{thm:mivlueSANIA}
Suppose the potential outcomes satisfy SANIA and the prior on the parameters has no correlation between units. 
If any unbiased estimators exist, then the weights for the MIV LUE are given by
\begin{equation}
    w_i(\mathbf{z}) =\frac{C_{i,d_i^{\mathbf{z}}}}{ \sum_{d=0}^{n-1} C_{i,d} } \frac{(2z_i-1)}{ n \Pr[\mathbf{z}^{obs}_i = z_i, d^{\mathbf{z}^{obs}}_i = d_i^{\mathbf{z}}]}
    \label{eq:mivlue_sania}
\end{equation}
where we define
\begin{align}
    C_{i,d} &=
 \left(   \sum_{\mathbf{z}} p(\mathbf{z}) \mathbf{I}\{d_i^{\mathbf{z}} = d\}  \frac{\Sigma(\mathbf{z})_{ii}  }{\Pr[z^{obs}_i=z_i, d^{\mathbf{z}^{obs}}_i=d]^2}  \right)^{-1}\\
 &= \left( \frac{\Sigma(0,d)_i}{\Pr[z_i^{obs}=0,d_i^{\mathbf{z}^{obs}}=d]} +\frac{\Sigma(1,d)_i}{\Pr[z_i^{obs}=1,d_i^{\mathbf{z}^{obs}}=d]}   \right)^{-1}.
\label{eq:HTscale}
\end{align}
or $C_{i,d}=0$ if either denominator in Eq.~\eqref{eq:HTscale} is zero.
\end{theorem}

We make a few observations about this estimator before directing our attention to the general case.
First, if $C_{i,d}=0$ for all $d$ then no unbiased estimators for $\bar{\beta}$ will exist as this violates the conditions of Proposition~\ref{prop:existSANIA}.
Second, for the case where the graph is empty so that SUTVA holds, then $C_{i,d}=0$ for each $d>0$ and this estimator simplifies to the Horvitz-Thompson estimator as stated in Theorem~\ref{thm:mivlueSUTVA}.

Finally, when the graph is non-empty the estimator is a weighted average of Horvitz-Thompson estimators.
Consider an the linear estimator $\hat{\beta}_i^{\tilde{w}}$ with weights $\tilde{w}_i$ defined as $w_i$ above but with $C_{i,d}$ defined to be $\mathbf{I}\{d=d_i^*\}$ for some $d^*$.
We require that $d_i^*$ satisfy $\Pr[z_i^{obs}=1|d_i^{\mathbf{z}^{obs}} = d_i^*]\in(0,1)$ for all $i$. 
Provided this holds for each $i$, the estimates $\hat{\beta}_i^{\tilde{w}}$ is unbiased for $\beta_i/n$ and hence $\sum_i \hat{\beta}_i^{\tilde{w}}$ is unbiased for $\bar{\beta}$.

The estimator in Theorem~\ref{thm:mivlueSANIA} is a weighted average over all valid values of $d_i^*$ of estimates like $\hat{\beta}_i^{\tilde{w}}$.
The value $C_{i,d}$ in Theorem~\ref{thm:mivlueSANIA} is equal to the inverse of the  integrated variance of $\hat{\beta}_i^{\tilde{w}}$ when $d_i^*=d$.
Hence, the estimator from Theorem~\ref{thm:mivlueSANIA} can be viewed as a weighted average of the estimators proposed in \citet{aronow2013estimating} with weights proportional to inverse integrated variances of said estimators. 




\subsubsection{General Prior and Constant Priors} 
\label{ssub:general_prior}
The general case in some ways resembles the case for the independent priors but does not have an explicit form.
Computation of the MIV LUE under a general prior requires finding the solution to a system of linear equations with $|\mathrm{supp}(\mathbf{z}^{obs})| n+2n+\sum_{i=1}^n |\mathrm{supp}(d_i^{\mathbf{z}^{obs}})\setminus \{0\}|$ unknowns.
If the covariance matrices are non-singular some simplifications can be made to reduce the computational complexity of computing this estimator. 
The details of for the general case are provided in Section~\ref{app:sania_general}.

As in Section~\ref{sub:mivlueSUTVA}, the opposite end of the spectrum from uncorrelated priors are priors where all parameters are constant across units.
Such simplifications unfortunately do not yield a closed form for the MIV LUE.
The next two examples illustrate the types of estimators that arise out of such constant parameter priors for two different graphs.



\begin{example}[Triangle plus tail]\label{ex:tri_w_tail_mivlue}
Consider again the graph from Example~\ref{ex:tri_w_tail_exist} and Figure~\ref{fig:tri_w_tail} with four nodes consisting of a triangle with one extra node adjacent to the first node in the triangle.
We consider the prior where with probability one $\alpha_i=\alpha$, $\beta_i=\beta$, and $\Gamma_i=\Gamma$ for all $i$. 
The prior distribution will be $\alpha\sim\Normal(0,1)$, $\beta\sim\Normal(0,1)$, and $\Gamma_i(d)\sim \Normal(0,d)$.
We consider the design which is a uniform sample over all allocations except for $\bm{0}$ and $\bm{1}$, essentially a Bernoulli trial with treatment probability 1/2 and the trivial allocations removed.
The resulting weights for the MIV LUE are given in Table~\ref{tab:tri_w_tail_mivlue}.
Careful investigation does not illuminate any straightforward pattern to these weights and moreover some weights seem nonsensical as high positive weights are given to units in the control group and visa versa.
Nonetheless, using prior distributions such as these will often lead to well performing estimators but this lack of a clear pattern is prototypical of MIV LUEs.
\qed
\end{example}

\begin{table}[ht]
\centering
\small
\begin{tabular}{r|rrr|rrr|rrr}
  \hline
$i$ & $z_i$ & $d_i^{\mathbf{z}}$ & $w_i(\mathbf{z})$ & $z_i$ & $d_i^{\mathbf{z}}$ & $w_i(\mathbf{z})$& $z_i$ & $d_i^{\mathbf{z}}$ & $w_i(\mathbf{z})$\\ 
  \hline
1 &   1 &   0 &   0 &   0 &   1 & -1.7 &   1 &   1 & 0.92 \\ 
  2 &   0 &   1 &  -2 &   1 &   0 & 1.3 &   1 &   1 & 0.41 \\ 
  3 &   0 &   1 &  -2 &   0 &   1 & 1.8 &   0 &   2 & -0.017 \\ 
  4 &   0 &   1 & 3.9 &   0 &   0 & -1.7 &   0 &   1 & -1.4 \\ 
   \hline
  \hline
$i$ & $z_i$ & $d_i^{\mathbf{z}}$ & $w_i(\mathbf{z})$ & $z_i$ & $d_i^{\mathbf{z}}$ & $w_i(\mathbf{z})$& $z_i$ & $d_i^{\mathbf{z}}$ & $w_i(\mathbf{z})$\\ 
  \hline
1 &   0 &   1 & -1.7 &   1 &   1 & 0.92 &   0 &   2 & 0.067 \\ 
  2 &   0 &   1 & 1.8 &   0 &   2 & -0.017 &   1 &   1 & 0.37 \\ 
  3 &   1 &   0 & 1.3 &   1 &   1 & 0.41 &   1 &   1 & 0.37 \\ 
  4 &   0 &   0 & -1.7 &   0 &   1 & -1.4 &   0 &   0 &  -1 \\ 
   \hline
  \hline
$i$ & $z_i$ & $d_i^{\mathbf{z}}$ & $w_i(\mathbf{z})$ & $z_i$ & $d_i^{\mathbf{z}}$ & $w_i(\mathbf{z})$& $z_i$ & $d_i^{\mathbf{z}}$ & $w_i(\mathbf{z})$\\ 
  \hline
1 &   1 &   2 & -2.5 &   0 &   1 & 0.13 &   1 &   1 & 1.4 \\ 
  2 &   1 &   2 & 1.4 &   0 &   0 & -0.85 &   0 &   1 & -0.69 \\ 
  3 &   1 &   2 & 1.4 &   0 &   0 & -0.85 &   0 &   1 & -0.69 \\ 
  4 &   0 &   1 & -0.23 &   1 &   0 & 1.3 &   1 &   1 & -0.015 \\ 
   \hline
  \hline
$i$ & $z_i$ & $d_i^{\mathbf{z}}$ & $w_i(\mathbf{z})$ & $z_i$ & $d_i^{\mathbf{z}}$ & $w_i(\mathbf{z})$& $z_i$ & $d_i^{\mathbf{z}}$ & $w_i(\mathbf{z})$\\ 
  \hline
1 &   0 &   2 & -0.18 &   1 &   2 & 1.4 &   0 &   2 & -0.18 \\ 
  2 &   1 &   0 & -0.49 &   1 &   1 & 0.4 &   0 &   1 & -0.37 \\ 
  3 &   0 &   1 & -0.37 &   0 &   2 & -1.4 &   1 &   0 & -0.49 \\ 
  4 &   1 &   0 & 1.3 &   1 &   1 & -0.45 &   1 &   0 & 1.3 \\ 
   \hline
  \hline
$i$ & $z_i$ & $d_i^{\mathbf{z}}$ & $w_i(\mathbf{z})$ & $z_i$ & $d_i^{\mathbf{z}}$ & $w_i(\mathbf{z})$\\ 
  \hline
1 &   1 &   2 & 1.4 &   0 &   3 &   0 \\ 
  2 &   0 &   2 & -1.4 &   1 &   1 & 0.064 \\ 
  3 &   1 &   1 & 0.4 &   1 &   1 & 0.064 \\ 
  4 &   1 &   1 & -0.45 &   1 &   0 & 0.43 \\ 
   \hline
\end{tabular}
\caption{ The (approximate) weights for the MIV LUE for each allocation in a Bernoulli trial excluding the two trivial allocations for Example~\ref{ex:tri_w_tail_mivlue}.}
\label{tab:tri_w_tail_mivlue}
\end{table}

Section~\ref{ssub:uncorrelated_prior}, Section~\ref{sub:mivlueNIA}, and Theorem~\ref{thm:mivlueSUTVA} illustrate how when the prior distribution has no correlation across units, the resulting MIV LUE uses inverse propensity weighting.
Theorem~\ref{thm:mivlueSUTVA} also states that under SUTVA and a constant parameter prior, the naive estimator is the MIV LUE provided the design is symmetric.
On the other hand, the previous example shows that under SANIA, neither the naive estimator nor a stratified version of it are the MIV LUE for the constant prior.
There are two key reasons that this fails in the previous example. 
In the following example, we illustrate a set of conditions under which the stratified naive estimator is the MIV LUE.
While these conditions are sufficient, we have not been able to show that they are necessary.

\begin{example}[Vertex Transitive Graphs]
One reason stratified naive estimators are not MIV LUEs, is they are not even unbiased.
The key reason that a naive-type estimator fails to be the MIV LUE in the previous example is the lack of sufficient symmetry.
Even if the graph was regular, so that all units have equal degree, stratified naive estimators are not necessarily MIV LUEs since the joint distribution of treated degrees and treatments is not exchangeable across units.

A sufficient property for the MIV LUE to be the stratified naive estimator is if the graph is vertex transitive and the design is sufficiently symmetric.
A graph is vertex transitive if for all pairs of units $i,j$, there exists a graph automorphism $\tau_{ij}:[n]\mapsto[n]$ such that $\tau(i)=j$ \citep{godsil2013algebraic}.\footnote{A graph automorphism $\tau:[n]\mapsto[n]$ is a bijection where $g_{ij}=g_{\tau(i)\tau(j)}$ for all $i,j$.}
Vertex transitive graphs include rings, hypercubes, and the Peterson graph and can be informally viewed as graphs which are locally identical at each vertex. 

Vertex transitivity is not quite sufficient to lead to the stratified naive estimator.
Recall that Theorem~\ref{thm:mivlueSUTVA}, part 2, requires that the design is symmetric and the support omits $\bm{0}$ and $\bm{1}$. 
For the SANIA case, we need to exclude any allocations where the set of treated degrees  for the control units, $\{d_i^{\mathbf{z}}: i\in [n], z_i=0\}$, is disjoint from the set of treated degrees for the treated units, $\{d_i^{\mathbf{z}}: i\in [n], z_i=1\}$. 
Viewing the treated degree as a covariate, this excludes any allocations which have no balance on the covariates. 

\begin{proposition}\label{prop:vertex_trans}
Suppose that the graph $g$ is vertex transitive and suppose the design $p$ satisfies,
\begin{enumerate}
    \item for any graph automorphism $\tau$ of $g$, $p(\mathbf{z})=p(\tau(\mathbf{z}))$ and
    \item for all allocations $\mathbf{z}$ where $\{d_i^{\mathbf{z}}: i\in[n], z_i=0\}\cap \{d_i^{\mathbf{z}}: i\in[n], z_i=1\}=\emptyset$, it holds that $p(\mathbf{z})=0$.
\end{enumerate}
Let $n_{z,d}=|\{i:z_i=z,d_i^{\mathbf{z}}=d\}|$.
For any prior where for all $i,j$, it holds that almost surely $\alpha_i=\alpha_j$, $\beta_i=\beta_j$, and $\Gamma_i=\Gamma_j$, the MIV LUE has weights given by
\begin{equation}
    w_i(\mathbf{z}) = \frac{C_{d_i^{\mathbf{z}}}(\mathbf{z})}{\sum_d C_d(\mathbf{z})} \frac{2z_i-1}{n_{z_i,d_i^{\mathbf{z}}}}, \label{eq:strat_naive}
\end{equation}
where 
\[
    C_d(\mathbf{z})=\mathbf{I}\{n_{0,d}>0, n_{1,d}>0 \} \left( \frac{1}{n_{0,d}}+\frac{1}{n_{1,d}} \right)^{-1}.
\]
\end{proposition}
The proof of this proposition is in Appendix~\ref{app:const_prior}.\qed
\end{example}


\subsection{NIA} 
\label{sub:mivlueNIA}

Under NIA it holds that $Y_i(\mathbf{z})=\alpha_i+\beta_i z_i + \Gamma_i(d_i^{\mathbf{z}}) = \Delta_i(d_i^{\mathbf{z}}) z_i$ .
The techniques to derive MIV LUEs are the same in the NIA case as in the SANIA case. 
One major difference is that the possible interaction between the direct treatment effect and the interference effect means that for an allocation $\mathbf{z}$ where $z_i=0$, it does not hold that $Y_i(\mathbf{z}+\mathbf{e}_i)-Y_i(\mathbf{z})=\beta$. 
This means that we cannot use stratified estimates like those derived for the SANIA case with uncorrelated priors.
Indeed, if only NIA or NIA with symmetrically received interference (SNIA) can be assumed and we find the MIV LUE under an uncorrelated prior then non-zero weights will only be assigned to units with no treated neighbors, as detailed in the following result.
\begin{theorem}
Under NIA (or SNIA), for a prior on the parameters which is uncorrelated across units, the MIV LUE has weights
\begin{equation}
   w_i(\mathbf{z}) = \frac{(2z_i-1)\mathbf{I}\{d_i^{\mathbf{z}}=0\}}{n\Pr[z^{obs}_i = z_i, d_i^{\mathbf{z}^{obs}}=0]}.
 \end{equation} 
\end{theorem}
Note that for correlated priors, the MIV LUE can have nonzero weights on units with positive treated degrees but again a closed form does not exist for the general case.

\subsection{SANASIA} 
\label{sub:mivlueSANASIA}

Consider the special case where the graph is undirected and not bipartite so that under SANASIA the potential outcomes can be parameterized as $Y_i(\mathbf{z}) = \alpha_i+\beta_i z_i +\gamma d^{\mathbf{z}}_i$.
In this case, we need to consider a prior distribution over the parameters $\alpha_1,\dotsc,\alpha_n$, $\beta_1,\dotsc,\beta_n$, and $\gamma$.
In the case of independent priors on all of these parameters, closed forms for the MIV LUE weights $\mathbf{w}(\mathbf{z})$ do exist but do not have a compact form. 
In Appendix~\ref{app:mivlue_sanasia}, we give the optimal weights in the independent prior case.


\section{Simulations}\label{sec:simulation}

We will now attempt a brief analysis via simulation of the quality of the estimators that we derived above in comparison to other commonly used estimators.
While a full exploration of the space of parameters is well beyond the scope of this work,
we will demonstrate how the performance of various estimators varies as a function of various aspects of the problem, including the potential outcome parameters, network properties, and the number of units.
While our theoretical results show that certain estimators will minimize the integrated variance among all unbiased estimators, under other notions of optimality certain estimators may be more desirable than others and our simulations will serve to investigate the impact of the choice of estimator in situations where the estimator is not necessarily optimal. 
In particular, since our focus in the theory is on finding estimators which minimize integrated variance, we will avoid integrating with respect to parameter distributions and instead focus on the performance of the estimators as the parameters change and with respect to the distributions of the networks.

Certain elements of the simulation will be common across experiments.
First, we consider six different estimators which we detail here.

\begin{description}
\item[Naive] Difference of means between the treatment group and the control group (see Eq.~\eqref{eq:est_naive}). The Naive estimator is the MIV LUE under SUTVA for the prior where all parameters are constant across units and if the design is invariant under permutation of the units. The Naive estimator is not unbiased if there is interference.
\item[Horvitz-Thompson] Inverse propensity score weighting ignoring interference (see Eq.~\eqref{eq:est_ht}). The Horvitz-Thompson estimator is the MIV LUE under SUTVA if the prior distribution is independent across units. This estimator is biased if there is interference.
\item[Stratified Naive] A weighted difference of means estimator stratified according to treated degree (see Eq.~\eqref{eq:strat_naive}). This estimator is MIV LUE under certain strong symmetry conditions as described in Proposition~\ref{prop:vertex_trans} and other this estimator can be biased.
\item[Independent] The MIV LUE under SANIA for the prior distribution where all parameters independent with $\alpha_i, \beta_i \sim \mathcal{N}(0,1)$ and $\Gamma_i(d)\sim \mathcal{N}(0,d)$ for all $i$ and $d$. (See Theorem~\ref{thm:mivlueSANIA}.)
\item[Equal] The MIV LUE under SANIA for a prior where $\alpha_i=\alpha+\epsilon_i$, $\beta_i=\beta$ and $\Gamma_i=\Gamma$ for all $i$.  
The parameters $\alpha,\beta$ and $\Gamma(1),\dotsc,\Gamma(d)$ are all independent with $\Gamma(d)\sim \mathcal{N}(0,d)$ and $\alpha,\beta\sim \mathcal{N}(0,1)$.
Additionally, $\epsilon_1,\dotsc, \epsilon_n \stackrel{iid}{\sim}\mathcal{N}(0,10^{-4})$.
\item[SANASIA] The MIV LUE under SANASIA, where all parameters are independent and $\alpha_i,\beta_i \sim \mathcal{N}(0,1)$ and $\gamma\sim \mathcal{N}(0,1/n)$.
\end{description}
These six estimators will be compared in various settings.
\begin{remark}
For the Equal estimator we impose some variance in the $\alpha_i$ parameters in order to guarantee that the $\Sigma(\mathbf{z})$ matrices are all non-singular. 
This allows for the much more rapid computational procedure as described in Section~\ref{sec:sania_nonsingular}. \qed
\end{remark}
\begin{remark}
For the SANASIA estimator, we choose $\gamma$ to have variance that is a factor of $1/n$ smaller than the variances for the $\Gamma$ parameters. 
Viewing SANASIA as an approximation for SANIA, $\gamma$ corresponds to an average across units of the $\Gamma$ parameters hence the variance of $\gamma$ matches the variance of $\bar{\Gamma}(1)$ for the Independent estimator. 
We could have scaled the variances for the Equal estimator but since the MIV LUE is invariant under uniform scaling of the variances, it would not change the results presented. \qed
\end{remark}

Computation times for these simulations will depend greatly on the two factors, the number of units and the number of allocations in the support of the design.
For most of the estimators, computation will be approximately linear in the number of units but for the Equal and SANASIA estimators the computation time will scale polynomially with the number of units since a matrix inversion is required. 

Additionally, since we will compute the mean square error with respect to the design, each of the estimators will have computation time which scales linearly with the number of allocations in the support of the design.
In order to consider more analyses we choose to keep these numbers small and to leave efforts towards more efficient computation for future work. 
To this end, the second common aspect of our simulations is that we will impose that the cardinality of the support is the minimum of $2^n-2$ and $2^{13}$, where the particular allocations used are sampled uniformly from all allocations.
Additionally, we always omit the allocations $\bm{0}$ and $\bm{1}$ where all units receive the same treatment.

Finally, for each of the simulations we impose that SANIA holds for the potential outcomes. 
While other sets of assumptions such as NIA may be of interest, restricting our attention to the SANIA setting allows us to investigate the impact of the core aspects of the problem, such as the number of units, the density and degree distribution of the network, and the effect sizes.
As we vary these aspects of the problem we will deviate from the regions of the parameter where each of the estimators tends to perform well, in order to discover how the performance of each estimator changes.

\subsection{Varying number of units and edges} 
\label{sub:vary_n_edge}

For this experiment we investigate the performance of the estimators as the number of units $n$ varies. 
We also consider two types of graphs, dense graphs, where for each Monte Carlo replicate we sample an \ER\ graph with edge probability $1/2$, and sparse graphs, where for each Monte Carlo replicate we sample an \ER\ graph with edge probability $1/n$.

For the potential outcomes, we impose that SANIA holds and use parameters that are common across all Monte Carlo replicates with the same number of units. 
The parameters are sampled independently with $\alpha_i\sim \mathcal{N}(0,1)$, $\beta_i \sim \mathcal{N}(2,1)$ and $\Gamma_i(d)=\mathcal{N}(d,1)$ using the same seed across replicates.

For each $n$ ranging in $10,20,\dotsc,50$, we sampled 500 graphs for each of the settings, dense $\mathrm{ER}(n,1/2)$ and sparse $\mathrm{ER}(n,1/n)$.
For each graph, we computed the six estimators and the mean square error for the estimators over the fixed randomization of the allocations $\mathbf{z}$.
We then averaged the mean square error across graphs with the same number of units for the dense and sparse settings.
These estimated mean square errors are plotted in Figure~\ref{fig:vary_n_density}.
We omit error bars as the number of Monte Carlo replicates were sufficiently large for the error bars to be quite small.

\begin{figure}[tb]
    \centering
    \includegraphics[width=\textwidth]{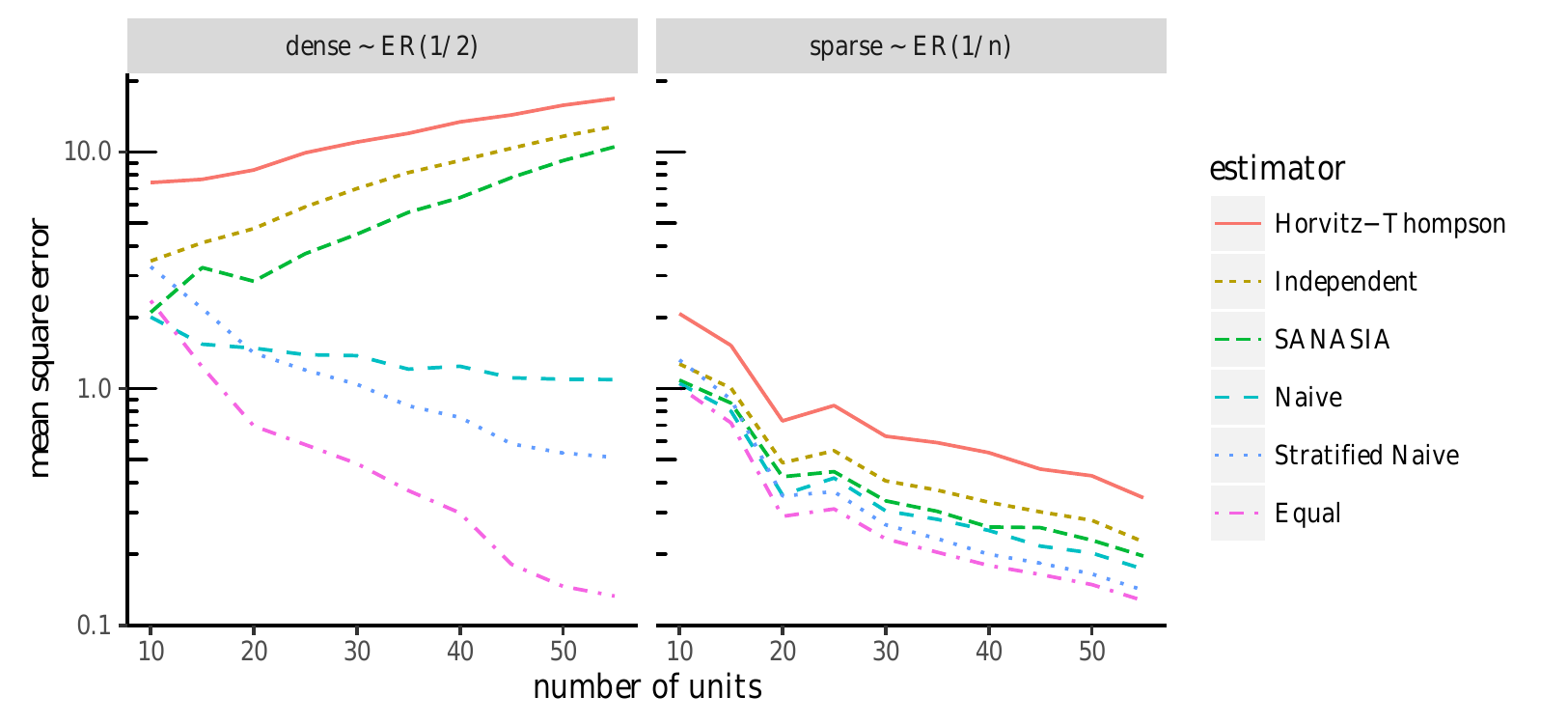}
    \caption{Mean square errors on the log scale for the six estimators as a function of the number of units for dense graphs and sparse graphs.}
    \label{fig:vary_n_density}
\end{figure}

First, ranks for the performance of each estimators are consistent across number of units and density.
Typically, Equal has the lowest mean square error followed by Stratified Naive, Naive, SANASIA, Independent, and finally Horvitz-Thompson.
However, in the case when the number of units is small and the graph is dense, this order can change.
Both the Horvitz-Thompson and Naive estimators are biased and do not take treated degree into account.
Nevertheless, Naive is generally one of the three best estimators according to these metrics.
The Horvitz-Thompson estimator has much higher variance and is always the worst estimator.

The MIV LUE Independent can be viewed as version of the Horvitz-Thompson estimator which stratifies based on treated degree.
Independent has slightly better performance, though the bias reduction does not make up for the large variance due to inverse propensity weighting.
The MIV LUE SANASIA is also derived from a prior where all the $\alpha_i$ and $\beta_i$ are independent but also imposes that the interference effects are constant and linear, which yields slight improvements over Independent.
The MIV LUE Equal imposes the most regularity in the prior, with all parameters being equal across units, which yields the best performance overall.

One problematic aspect of Horvitz-Thompson, Independent, and SANASIA, each of which implicitly or explicitly consider heterogeneity in baseline outcomes and direct effects, is that in the dense case, the mean square error actually increases with the number of units.
This suggests that these estimators have difficulty accounting for large interference effects and indeed, unlike naive type estimators, Horvitz-Thompson estimators can have variances that scale with effect sizes.

All the estimators improve as the number of units increase in the sparse case, where the overall interference effect per unit remains approximately constant in the number of units.

\subsection{Varying effect size}

In this next experiment, we investigate the impact of the size of the effect on the performance of the estimators. 
This is of interest since the MIV LUEs are each derived from mean zero priors so we might expect performance to vary as the effect sizes increase. 

For each monte carlo replicate in this setting we sampled an \ER\ random graph with edge probability $1/2$ and $12$ units. 
The parameters for the potential outcomes were sampled independently with $\alpha_i\sim \mathcal{N}(0,1)$, $\beta_i \sim \mathcal{N}(\mu_\beta,1)$ and $\Gamma_i(d)\sim \mathcal{N}(\mu_\gamma d, 1)$. 
The mean direct treatment effect (DTE) size $\mu_\beta$ ranged in $0,2,4$ and the interference effect (IE) size ranged in $0,0.2,\dotsc, 1.8$.

\begin{figure}[tb]
    \centering
    \includegraphics[width=\textwidth]{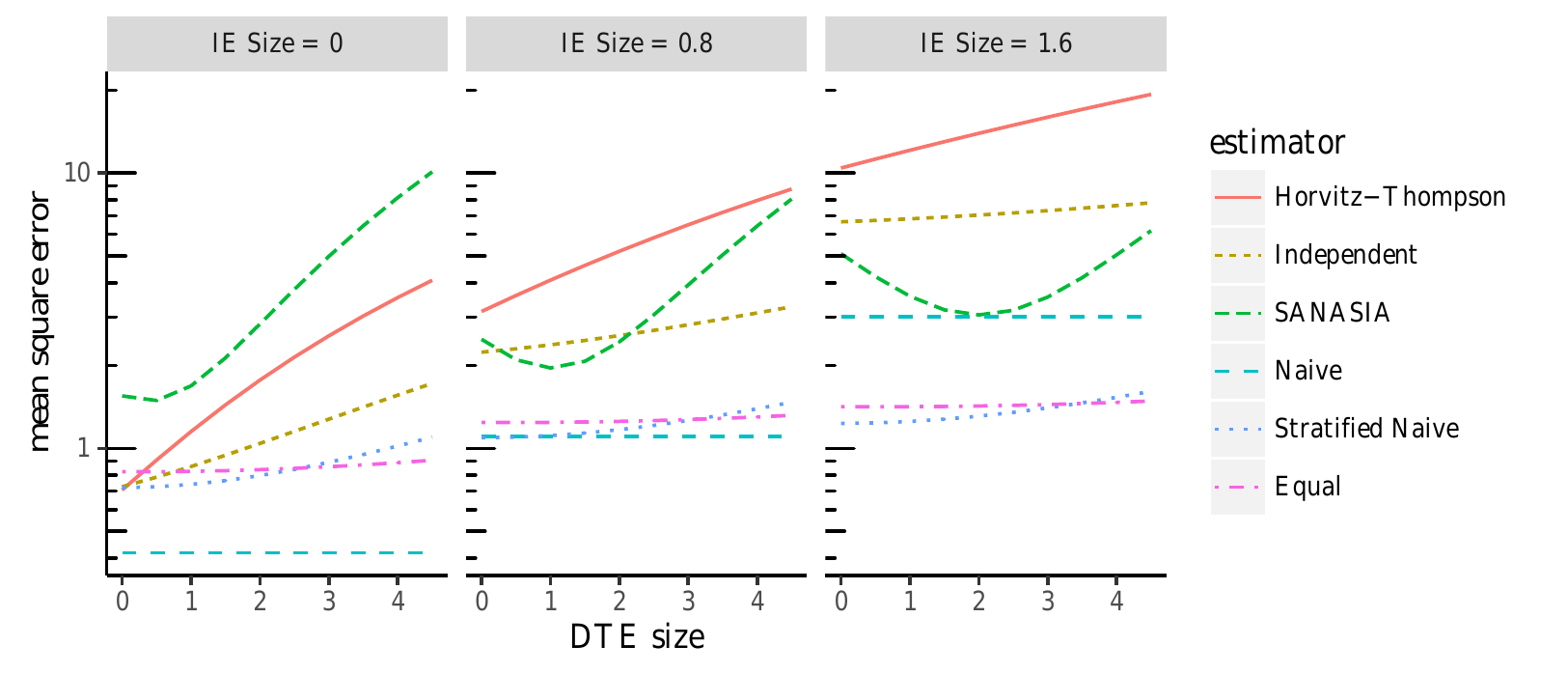}
    \caption{Mean square errors on the log scale for the six estimators as a function of the number of the direct treatment and interference effect sizes.}
    \label{fig:vary_beta_gamma}
\end{figure}

Figure~\ref{fig:vary_beta_gamma} shows the mean square errors for the six estimators described above in these settings. 
First, the mean square error for the naive estimate is invariant to shifts in the DTE size yielding flat lines in each panel.
The stratified naive estimator is not quite flat due to the fact that the bias of this estimator increases slowly with the DTE size due to the lack of balance in the treated degrees for each node.
The MIV LUE equal is also very close to having mean square error invariant to effect size, though it does increase slowly.

On the other hand, both the Horvitz-Thompson and Independent estimators have variances that depend strongly on effect sizes and hence each become poor estimators quite quickly as these increase. 
Each estimator, including the MIV LUEs, have mean square errors which increase with interference effect sizes as the estimators have more to compensate for. 

The MIV LUE SANASIA behaves quite differently from the other estimators, with the mean square error first decreasing with the DTE size and then increasing rather quickly. 
We have not fully investigated the behavior of this estimator but we believe that tuning this estimator more carefully may yield significant improvements.

Overall, we again see that when inteference effects do not exist, a large penalty is payed by using estimators that account for them, but when interference effects are present both the Equal and the Stratified Naive estimates performing well, with the Equal estimators less prone to bias when the direct effect is large.

We also investigated the impact of treatment effect heterogeneity through the variance of the $\beta_i$ and $\Gamma_i(d)$ terms, but we found that the impact of treatment effect heterogeneity was small compared to the impact of effect sizes.
The optimal estimators tended to not change quickly in terms of the level of treatment effect heterogeneity.
While it may be natural to posit that an estimator such as Independent will tend to have relatively superior performance as the effect heterogeneity increases, we found that this is only marginally the case. 
The Independent was less impacted by treatment effect heterogeneity, with a relatively small slope as effect heterogeneity increased, however, the overall scale of the errors for Independent leads to the fact that the effect heterogeneity must be very large before Independent will perform better than other estimators we investigated.
Additionally, the naive estimator was the least susceptible to treatment effect heterogeneity and hence is frequently the best estimator when treatment effect variances are large. 
The performance of the estimators could be impacted by other aspects of the effect size distribution such as the skewness but we have not yet investigated the nature of those impacts.





\subsection{Varying Degree Distribution}

For the simulations above, we only considered graphs distributed according \ER\ distributions which many have noted do not adequately fit most real world graphs \citep{barabasi1999emergence,newman2003structure}
For example, the degree distribution of \ER\ graphs are approximately Poisson or Binomial while many real world graphs exhibit more skewed degree distributions such as power law distributions.
In this simulation we generate graphs with different degree distributions using variations on the preferential attachment model \citep{barabasi1999emergence}. 
In the preferential attachment model, nodes are added one at a time, with new nodes having one edge connecting it to one of the nodes currently in the graph.
The probability that the edge connects to a node currently in the graph is proportional to $d_i^{\alpha}$, where $d_i$ is the current degree of node $i$ and $\rho$ is a parameter.
Once all $n$ nodes have been added and attached to the graph, we randomly permute the vertices, which preserves the degrees but not the order that the nodes were added.
This ensures that the fixed designs are not biased towards higher or lower degree nodes.

In the standard preferential attachment model, $\rho=1$ which yields a power law distribution for the degrees.
When $\rho=0$, the distribution is similar to an \ER\ distribution, while when $\rho=2$, the degree distribution becomes highly skewed, frequently yielding a star graph.

\begin{figure}[tb]
    \centering
    \includegraphics[width=4.4in]{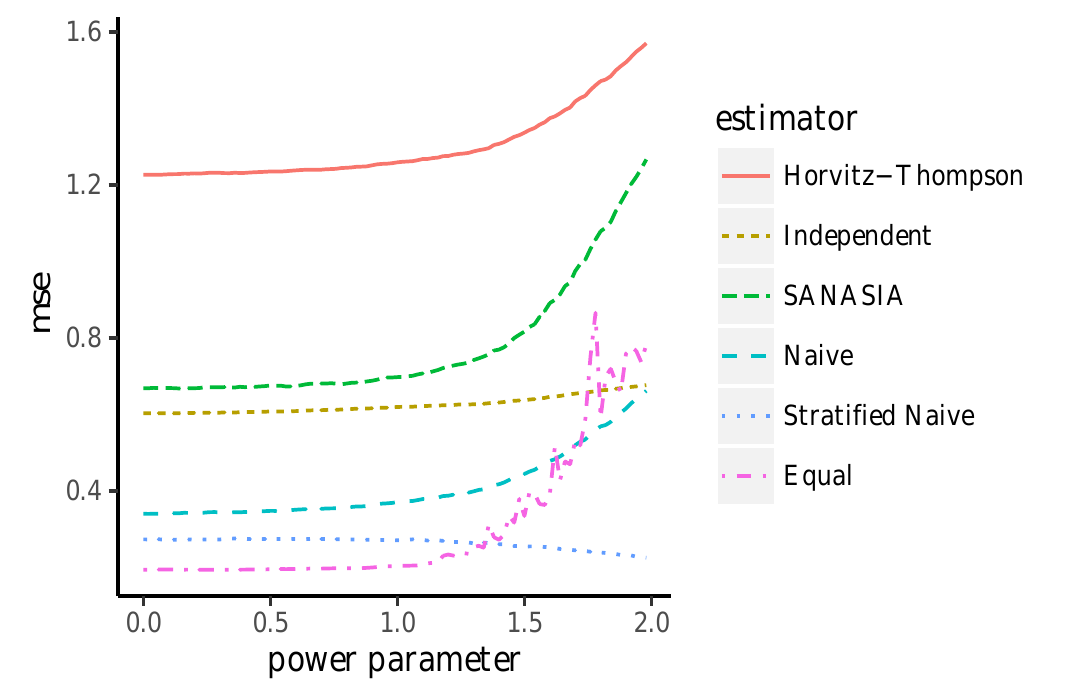}
    \caption{Mean square errors for the six estimators as a function of the number of the direct treatment and interference effect sizes.}
    \label{fig:vary_degree_dist}
\end{figure}

Figure~\ref{fig:vary_degree_dist} shows the mean square error for as a function of the power parameter $\rho$. 
We see that most estimators have there performance degrade increasingly quickly as the power parameter increases, especially beyond $\rho=1.5$.
The best estimator up to that point is again the MIV LUE Equal.
This estimator has performance that becomes non-smooth when the power gets large. We believe this occurs because the distribution of the graphs tends to concentrate and the Equal estimator may be sensitive to small changes in the graph structure. 
Finally, the stratified naive estimator actually has improved performance as the degree distribution becomes more skewed.
This is likely due to the fact that the stratified naive estimator will omit units that have high treated degree, effectively estimating the mean DTE for the remaining nodes which all have similar treated degrees.
Estimators such as Equal instead must be fully unbiased and hence cannot omit nodes with large treated degree.

\section{Discussion}\label{sec:disc}

In this work we have considered the neighborhood interference assumption and four core assumptions structural assumptions regarding the potential outcomes. 
The neighborhood interference assumption is the simplest nonparametric assumption for which units' treatments' can effect which units' outcomes' in a setting with networked units.
As we have seen, even this assumption can lead to complex high-dimensional parameterizations for the potential outcomes that make estimation and inference difficult.

Some practitioners may desire to allow for units at a greater distance to effect each other and this can be done but it requires either allowing for even more potential outcomes, losing information about path lengths by rewiring the graph, or parameterizing how vertices at a greater distance can effect units.
Additionally, due to the small world property of most observed networks, allowing treatments to effect units at a distance of more than two or three edges away would often be nearly equivalent to the arbitrary interference setting.
Hence, even though NIA does not capture all aspects of possible interference in a real problem, we believe for many practitioners NIA does provide a sufficiently rich structure to approximate real-world interference, while promoting a more parsimonious parameterization.

A more challenging problem for many real world setting will be that the observed network does not fully capture the possible paths of interference.
If we view the observed network as a noisy version of a true network encoding interference, then under SANIA, the observed treated degree will also be a noisy observation of the true treated degree for that allocation.
In these settings, it may be necessary to account for this noise by modelling the true network and the noise structure for the observed network or by modelling the noisy structure of the interference directly.
We believe these problems pose very interesting challenges even in the SANIA setting.

We focused on estimation of the direct treatment effect and most of our efforts were restricted to the SANIA setting.
Experimenters may be interested in other estimands under other sets of assumptions but we believe that the general structure of our results regarding how to chose among unbiased estimates will hold for other estimands provided the constraints for unbiased are changed appropriately.

Another core aspect of this work is to consider unbiased estimators, and while we showed in simulation that this can lead to improved performance it still is not clear that unbiased estimators should be generally preferred. 
Our simulations showed that when interference was present, the Equal MIVLUE estimator and the stratified naive estimator performed comparably.
Unbiasedness may be desired if combining the results of many experiments is a goal and unbiased estimators can sometimes be more easily interpreted.
In general, the desired optimality properties of the estimator must be chosen by the practioner.
If unbiasedness is wanted, we have shown that using more generic unbiased estimators based on inverse propensity score weighting can lead to a significant sacrifice in performance. 
Additionally, the minimum integrated variance concept provides a core tool for choosing among unbiased estimators. 

One of the major challenges for using these ideas on large scale social networks is the computational burden for computing the MIVLUEs, especially those where the prior imposes correlation between units.
For relatively simple designs such as Bernoulli trials and completely randomized designs, computation of propensities is staightforward and so for uncorrelated priors the computation can be easily paralyzed across units.
For correlated priors, computation of $\Sigma(\mathbf{z})^{-1}$ is required for each $\mathbf{z}$ in the support of $p$. 
This can be parallelized across allocations with the relevant $\mathbf{w}(\mathbf{z})$ computed in a MapReduce framework.

To reduce the computational burden, practitioners do have a few options.
One option is to reduce the number of units or repeat the experiment across small disconnected subgraphs.
Another option, which we adopted in the simulations, is to reduce the support of the design $p$.
While we did this randomly, one could use rerandomization to reduce the support towards certain desirable designs, such as those with low treated degrees or balanced treated degrees \citep{basse2015optimal}.

In general, constructing the design presents many challenges and one promising aspect of this work is to use the minimum integrated variance ideas to motivate design decisions as well as estimation procedures.
Ideally, one could seek to jointly optimize the design $p$ and the weights $\mathbf{w}$ but this introduces further computational challenges since an iterative approach could involve finding a MIV LUE at each iteration.
Anoter option is to first optimize the design based on other criteria such as balance for the treated degree and then find the MIV LUE for that fixed design.

Finally, we considered the relatively unrealistic setting where there are no additional covariates about the units.
Unit covariates can be handled by many classical tools from causal inference and tools like MIVLUEs can also be applied in analogous manner.
A challenge more closely related to this manuscript is if the edges between units have covariate information.
If this information is categorical, such as a communication type or topic, then then one may want to partition the treated degree according to each edge type, effectively increasing the types of interference treatments.
If these covariates are continuous or more complex then accounting for that information presents further challenges.





\section*{Acknowledgements}
This work was partially supported by National Science Foundation awards CAREER IIS-1149662 and IIS-1409177, and by Office of Naval Research YIP award N00014-14-1-0485, all to Harvard University.  EMA is an Alfred P. Sloan Research Fellow, and a Shutzer Fellow at the Radcliffe Institute for Advanced Study.

\bibliography{../refs/bibliography}

\appendix


\section{Parametric Forms for Combinations of Assumptions}\label{app:param}

Here we give parameterizations for the remaining nine combinations of the four assumptions from Section~\ref{sec:structural}.

\begin{definition}[SNIA]
If we assume symmetrically received interference effects then the potential outcomes can be parameterized as 
\begin{equation}
    Y_i(\mathbf{z}) = \alpha_i + \beta_i z_i + \Gamma_i(d_i^{\mathbf{z}}) + \Delta_i(d_i^{\mathbf{z}})z_i ,
\end{equation}
where $\alpha_i,\beta_i\in\Re$ and $\Gamma_i,\Delta_i:\{0,\dotsc,d_i\}\mapsto\Re$.
\end{definition}

\begin{definition}[ANIA]
The potential outcomes satisfy additivity of main effects if and only if the the potential outcomes can be parameterized as 
\begin{align*}
Y_i(\mathbf{z}) = \alpha_i +\beta_i z_i + \Gamma_i(\mathbf{z}_{\mathcal{N}_i}) 
\end{align*}
where $\alpha_i,\beta_i\in\Re$ and $\Gamma_i:\{0,1\}^{d_i}\mapsto \Re$.
\end{definition}

\begin{definition}[NAIA]
The potential outcomes satisfy additivity of interference effects if and only if the potential outcomes can be parameterized as
 \begin{equation}
    Y_i(\mathbf{z}) = \alpha_i + \beta_i z_i + \sum_{j} \gamma_{ji} g_{ji} z_j +  \sum_{j} \delta_{ji} g_{ji} z_i z_j 
 \end{equation}
 where for each $i,j\in[n]$, $\alpha_i,\beta_i,\gamma_{ij},\delta_{ij}\in \Re$.
 \end{definition} 
\begin{remark}
Each $\gamma_{ji}$ represents the effect of unit $j$ being treated on the outcome of unit $i$ provided there is an edge from $j$ to $i$ and
similarly, $\delta_{ij}$ represents the additional effect due to any interaction between  the treatments of $i$ and $j$.
Note that for each unit $i$ there are still $2^{d_i+1}$ distinct potential outcomes but those outcomes can be parametrized using only $2d_i+2$ parameters. \qed
\end{remark}

\begin{definition}[ANAIA]
The potential outcomes satisfy additivity of main effects and additivity of interference effects if and only if the potential outcomes can be parameterized
\begin{equation}
  Y_i(\mathbf{z}) = \alpha_i+\beta_i z_i + \sum_{j} \gamma_{ij} g_{ij} z_j
\end{equation}
for  $\alpha_i,\beta_i,\gamma_{ij}\in \Re$.
\end{definition}

\begin{definition}[SNAIA]
The potential outcomes satisfy symmetrically received interference effects and additivity of interference effects if and only if the potential outcomes can be parameterized
\begin{align}
  Y_i(\mathbf{z}) &= \alpha_i + \beta_i z_i + \gamma_{i} d_i^{\mathbf{z}} +   \delta_{i} d_i^{\mathbf{z}} z_i
\end{align}
for some $\alpha_i,\beta_i,\gamma_i,\delta_i\in\Re$.
\end{definition}

\begin{definition}[SANAIA]
The potential outcomes satisfy symmetrically received interference effects, additivity of main effects, and additivity of interference effects if and only if the potential outcomes can be parameterized
\begin{align}
  Y_i(\mathbf{z}) &= \alpha_i + \beta_i z_i + \gamma_{i} \sum_{j} g_{ji} z_j \\
  &= \alpha_i + \beta_i z_i + \gamma_{i} d_i^{\mathbf{z}}
\end{align}
for some $\alpha_i,\beta_i,\gamma_i\in \Re$.
\end{definition}




\begin{definition}[NASIA]
The potential outcomes satisfy additivity of interference effects and symmetrically sent interference effects if and only if the potential outcomes can be parameterized
\begin{equation}
   Y_i(\mathbf{z}) = \alpha_i + \beta_i z_i + \sum_{j} \gamma_j g_{ji} z_j +  \sum_{j} \delta_j g_{ji} z_i z_j 
\end{equation}
for some $\alpha_i, \beta_i, \gamma_j, \delta_j\in \Re$.
\end{definition}

\begin{definition}[ANASIA]
The potential outcomes satisfy additivity of main effects, additivity of interference effects, and symmetrically sent interference effects  if and only if the potential outcomes can be parameterized
\begin{equation}
   Y_i(\mathbf{z}) = \alpha_i + \beta_i z_i + \sum_{j} \gamma_j g_{ji} z_j 
\end{equation}
for some $\alpha_i, \beta_i, \gamma_j\in \Re$.
\end{definition}

\begin{definition}[SNASIA]
The potential outcomes satisfy symmetrically received interference effects, additivity of interference effects, and symmetrically sent interference effects  if and only if the potential outcomes can be parameterized
\begin{equation}
   Y_i(\mathbf{z}) = \alpha_i + \beta_i z_i + \gamma_i d_i^{\mathbf{z}} + \delta_i d_i^{\mathbf{z}}.
\end{equation}
for some $\alpha_i, \beta_i, \gamma_i, \delta_i\in \Re$, where as in Proposition~\ref{prop:paramSANASIA}, $\gamma_i=\gamma_j$ and $\delta_i=\delta_j$ if $i$ and $j$ are in the same connected component of $h=\mathbf{I}\{g^T g>0\}$.
\end{definition}



\section{MIV LUE under SANIA}
In this section we will develop a method to solve for the minimum integrated variance linear unbiased estimated under SANIA for quite general priors. 
The method relies on the use of Lagrange multipliers to solve the constrained optimization problem \citep{boyd2004convex}
We will require that $\Sigma(\mathbf{z})$ is positive semidefinite for each $\mathbf{z}$ in the support of the design. 
Note, we do not require $\Sigma(\mathbf{z})$ to be nonsingular, since singular $\Sigma(\mathbf{z})$ arise for certain priors, such as those where parameters are almost surely constant across units. 

\subsection{General Case}\label{app:sania_general}

Using the definition of $\Sigma(\mathbf{z})$ from Eq.~\eqref{eq:sigma_z}, half of the integrated mean square error for an estimate $\hat{\beta}^{\mathbf{w}}$ is given by \[
   \frac{1}{2} \sum_{\mathbf{z}} p(\mathbf{z}) \mathbf{w}(\mathbf{z})^T \Sigma(\mathbf{z}) \mathbf{w}(\mathbf{z}).
\]
We use half the IMSE because it simplifies some formulas later without changing the solution. 

Using the method of Lagrange multipliers we know that the best weights will be a stationary point function
\begin{align*}
L(\mathbf{w},\lambda) = \quad \frac{1}{2}&\sum_{\mathbf{z}} p(\mathbf{z}) \mathbf{w}(\mathbf{z})^T \Sigma(\mathbf{z}) \mathbf{w}(\mathbf{z}) \\
+ \lambda_\alpha^T &\sum_{\mathbf{z}} p(\mathbf{z}) \mathbf{w}(\mathbf{z}) \\
+ \lambda_\beta^T & \left( \sum_{\mathbf{z}} p(\mathbf{z}) \mathbf{w}(\mathbf{z}) \mathbf{z} - \frac{1}{n} \bm{1} \right) \\
+ \sum_{d=1}^{n-1} \lambda_{\Gamma d}^T  &\sum_{\mathbf{z}} p(\mathbf{z}) \mathbf{w}(\mathbf{z}) \mathbf{I}\{d^{\mathbf{z}} =d\}.
\end{align*}
 where $\mathbf{w}:\{0,1\}^n \mapsto \Re^n$ and $\lambda_\alpha,\lambda_\beta,\lambda_{\Gamma 1},\dotsc, \lambda_{\Gamma n-1}\in \Re^n$ are vectors.
 Note that we use the notation $\mathbf{I}\{d^{\mathbf{z}} =d\}$ to refer to the vector of indicators, so that $\mathbf{I}\{d^{\mathbf{z}} =d\}_i=\mathbf{I}\{d^{\mathbf{z}}_i =d\}$.
 Similar abuses of notation will be made later, hopefully with minimum confusion for the reader.
 Finally, for a discrete random variable $X$, its support is the set of values which the random variable takes with some positive probability which we denote $\mathrm{supp}(X)= \{x: \Pr[X=x]>0\}$.

In order to find the MIV LUE, we must find the Karus-Kuhn-Tucker (KKT) point \citep{boyd2004convex} for this problem which is the pair $\mathbf{w},\lambda$ which satisfies 
\[
    \nabla L \begin{pmatrix}
        \mathbf{w} \\ \lambda
    \end{pmatrix}  = \bm{0}.
\]
This is a system of linear equations with $|\mathrm{supp}(\mathbf{z}^{obs})| n+2n+\sum_{i=1}^n |\mathrm{supp}(d_i^{\mathbf{z}^{obs}})\setminus \{0\}|$ unknowns and coefficients.
As is the case for generalized least squares estimates for linear regression, $\nabla L$ comprises $2\times 2$ blocks where the bottom right block of size $2n+\sum_{i=1}^n |\mathrm{supp}(d_i^{\mathbf{z}^{obs}})\setminus \{0\}|$ is zero.

\begin{remark}
For a Bernoulli trial, $\nabla L$ has $n 2^n+2n+\sum_i d_i$ rows which for $n=12$ is nearly 50000. 
Fortunately, this system is quite sparse, with only $O(2^n n^2)$ non-zero entries, rather than $O(4^n n^2)$, so the use of sparse linear system solvers may lead to faster algorithms \citep{saad2003iterative}.
Furthermore, the support of the design can be restricted, as was done in our simulations, to ease this computational burden.
\end{remark}

\subsection{Nonsingular case} \label{sec:sania_nonsingular}
If it holds that $\Sigma(\mathbf{z})$ is nonsingular for each $\mathbf{z}$ in the support of $p$, then this system can be solved more rapidly.
First, taking derivatives with respect to $\mathbf{w}(\mathbf{z})$ for a given unit $i$ and allocation $\mathbf{z}$, we have that 
\[
  \frac{\partial L}{\partial \mathbf{w}(\mathbf{z})} = p(\mathbf{z}) \left(\Sigma(\mathbf{z}) \mathbf{w}(\mathbf{z}) -  \lambda_\alpha  - \mathrm{diag}(\mathbf{z}) \lambda_\beta  - \sum_{d=1}^{n-1} \mathrm{diag}(\mathbf{I}\{d^{\mathbf{z}} =d\})\lambda_{\Gamma d}\right)
\]
As $\Sigma(\mathbf{z})$ is invertible, we can solve for $\frac{\partial L}{\partial \mathbf{w}(\mathbf{z})}=\bm{0}$. 
This yields
\begin{equation}
    \mathbf{w}(\mathbf{z}) = - \Sigma(\mathbf{z})^{-1} \left( \lambda_{\alpha}+\mathrm{diag}(\mathbf{z}) \lambda_{\beta} + \sum_{d=1}^{n-1} \mathrm{diag}(\mathbf{I}\{d^{\mathbf{z}} = d\}) \lambda_{\Gamma d}\right). \label{eq:wz_inv}
\end{equation}

Taking derivatives with respect to the $\lambda$ terms and plugging in Eq.~\eqref{eq:wz_inv} gives
\begin{align*}
\frac{\partial L}{\partial \lambda_\alpha} &= - \sum_{\mathbf{z}} p(\mathbf{z}) \Sigma(\mathbf{z})^{-1} \left(  \lambda_\alpha  + \mathrm{diag}(\mathbf{z}) \lambda_\beta  + \sum_{d=1}^{n} \mathrm{diag}(d^\mathbf{z}=d)\lambda_{\Gamma d}\right)\\
\frac{\partial L}{\partial \lambda_\beta} &= - \sum_{\mathbf{z}} p(\mathbf{z})  \mathrm{diag}(\mathbf{z}) \Sigma(\mathbf{z})^{-1} \left( \lambda_\alpha  + \mathrm{diag}(\mathbf{z}) \lambda_\beta  + \sum_{d=1}^{n} \mathrm{diag}(d^\mathbf{z}=d)\lambda_{\Gamma d}\right)- \frac{1}{n}\bm{1}\\
\frac{\partial L}{\partial \lambda_{\Gamma d}} &= -\sum_{\mathbf{z}} p(\mathbf{z})  \mathrm{diag}(d^\mathbf{z}=d)  \Sigma(\mathbf{z})^{-1} \left( \lambda_\alpha  + \mathrm{diag}(\mathbf{z}) \lambda_\beta  + \sum_{d'=1}^{n} \mathrm{diag}(d^\mathbf{z}=d')\lambda_{\Gamma d'}\right)
\end{align*}
To find the optimal $\lambda$, we must set the above equations to zero and solve for the $\lambda$ terms. 
This is now a system with $2n+\sum_{i} |\mathrm{supp}(d_i^{\mathbf{z}}))\setminus \{0\}|$ unknowns and coefficients which must have a solution since we are working under the conditions where an unbiased estimator exists.
Finally, after solving for the $\lambda$ terms the MIV LUE is found by plugging them back into Eq.~\eqref{eq:wz_inv}.
In the next section we will show how we can analytically solve this system in the case of uncorrelated priors which will lead to Theorem~\ref{thm:mivlueSANIA}.

\begin{remark}\label{rem:compute}
The case when $\Sigma(\mathbf{z})$ is invertible offers a substantial reduction in computation compared to the general case.
First, we must find all $\Sigma(\mathbf{z})^{-1}$ involves inverting $|\mathrm{supp}(\mathbf{z}^{obs})|$ matrices of size $n\times n$. 
Following this, we need to sum various diagonal transformations of the $\Sigma(\mathbf{z})^{-1}$ over all allocations.

We note that these two operations fit well into a MapReduce framework and hence can be performed relatively quickly in a distributed or parallel system.
At this point we have not fully explored the computation aspects of this procedure but further methodological research will likely allow scaling to relatively large problems provided the the support of the design does not grow too quickly.\qed
\end{remark}

\begin{remark}
If $\Sigma(\mathbf{z})$ is singular, the strategy outlined in this section can still be attempted with the inverse of $\Sigma(\mathbf{z})$ replaced by the Moore-Penrose pseudoinverse. 
Unfortunately, this method is not guaranteed to find an optimal solution but it is straightforward to check whether the KKT conditions are satisfied.
Again, we have not explored this in detail but we do note that this method finds the solution found in Section~\ref{app:const_prior}.\qed
\end{remark}



\subsection{Uncorrelated Case}
In this subsection, we'll prove Theorem~\ref{thm:mivlueSANIA}.
Uncorrelated priors will automatically yield non-singular covariance structure, so everything in the previous subsection applies in this case but the equations simplify substantially.

\begin{proof}[Proof of Theorem~\ref{thm:mivlueSANIA}]
First, $\Sigma(\mathbf{z})\in \Re^{n\times n}$ will be diagonal and with $\Sigma(\mathbf{z})_{ii}=\sigma_i(\mathbf{z})$, 
\begin{equation}
    w_i(\mathbf{z}) = - \sigma_i(\mathbf{z})^{-1} \left( \lambda_{\alpha i}+z_i \lambda_{\beta i} + \sum_{d=1}^{n-1} \mathbf{I}\{d^{\mathbf{z}}_i = d\} \lambda_{\Gamma d,i}\right). \label{eq:wz_uncorr}
\end{equation}
Let
\begin{align}
a_i &= \sum_{\mathbf{z}} \frac{p(\mathbf{z})}{\sigma_i(\mathbf{z})} &b_i &= \sum_{\mathbf{z}} \frac{p(\mathbf{z}) z_i }{\sigma_i(\mathbf{z})} \label{eq:ai_bi}\\
f_{id} &= \sum_{\mathbf{z}} \frac{p(\mathbf{z}) \mathbf{I}\{d_i^{\mathbf{z}} = d\} }{\sigma_i(\mathbf{z})}  & g_{id} &= \sum_{\mathbf{z}} \frac{p(\mathbf{z}) z_i \mathbf{I}\{d_i^{\mathbf{z}} = d\} }{\sigma_i(\mathbf{z})} \label{eq:fi_gi}
\end{align}
so that 
\begin{align*}
   \frac{\partial L}{\partial \lambda_{\alpha i}} &= - \left(\lambda_{\alpha i} a_i + \lambda_{\beta i} b_i + \sum_{d=1}^{n-1} \lambda_{\Gamma,d,i}  f_{id}\right),\\
   \frac{\partial L}{\partial \lambda_{\beta i}} &=  -\left(\lambda_{\alpha i} b_i +\lambda_{\beta i } b_i +  \sum_{d=1}^{n-1} \lambda_{\Gamma,d,i} g_{id}  + 1/n\right), \text{ and }\\ 
   \frac{\partial L}{\partial \lambda_{\Gamma,d,i}} &=-\left(\lambda_{\alpha i} f_{id}+\lambda_{\beta i }g_{id} + \lambda_{\Gamma,d,i}f_{id}\right).
\end{align*}

The last equation implies $\lambda_{\Gamma,d,i}=-\lambda_{\alpha i} - \lambda_{\beta i} \frac{g_{id}}{f_{id}}$ which we plug into the first two equations yielding

\begin{align*}
   \frac{\partial L}{\partial \lambda_{\alpha i}} &=  \lambda_{\alpha i} \left(-a_i+\sum_{d=1}^{n-1}f_{id}\right) + \lambda_{\beta i} \left(-b_i+\sum_{d=1}^{n-1} g_{id}\right),\\
   \frac{\partial L}{\partial \lambda_{\beta i}} &=  \lambda_{\alpha i} \left( -b_i +\sum_{d=1}^{n-1} g_{id}\right) +\lambda_{\beta i } \left( -b_i  + \sum_{d=1}^{n-1} \frac{g_{id}^2}{f_{id}} \right) - 1/n.
\end{align*}

Setting the above to 0 gives 
\begin{align*}
\lambda_{\alpha i} &= \frac{g_{i0}/f_{i0}}{n \sum_{d=0}^{n-1}g_{id} \left( 1- \frac{g_{id}}{f_{id}} \right) } \\
\lambda_{\beta i} &= \frac{-1}{n \sum_{d=0}^{n-1}g_{id} \left( 1- \frac{g_{id}}{f_{id}} \right) }\text{ and} \\
\lambda_{\Gamma d i} &= \frac{-g_{i0}/f_{i0}+g_{id}/f_{id}}{n \sum_{d=0}^{n-1}g_{id} \left( 1- \frac{g_{id}}{f_{id}} \right) }.
\end{align*}
Which implies after plugging into  Eq.~\eqref{eq:wz_uncorr} that 
\[
  w_i(\mathbf{z}) = \frac{z_i-g_{id_i^{\mathbf{z}}}/f_{id_i^{\mathbf{z}}}}{n\sigma_i(\mathbf{z}) \sum_{d=0}^{n-1}g_{id} \left( 1- \frac{g_{id}}{f_{id}} \right) }.
\]
By letting $h_{id}=f_{id}-g_{id}$, we can rewrite the above equation as 
\begin{align*}
  w_i(\mathbf{z}) &= \dfrac{z_i  \left( \frac{g_{id_i^{\mathbf{z}}} h_{id_i^{\mathbf{z}}} }{f_{id_i^{\mathbf{z}}}} \right)  }{n\sigma_i(\mathbf{z}) g_{id_i^{\mathbf{z}}} \sum_{d=0}^{n-1} \frac{g_{id}h_{id}}{f_{id}}}-\dfrac{(1-z_i)  \left( \frac{g_{id_i^{\mathbf{z}}} h_{id_i^{\mathbf{z}}} }{f_{id_i^{\mathbf{z}}}} \right)  }{n\sigma_i(\mathbf{z}) h_{id_i^{\mathbf{z}}} \sum_{d=0}^{n-1} \frac{g_{id}h_{id}}{f_{id}}} \\
  &=\left(  n \Pr[z_i^{obs}=z_i,d_i^{\mathbf{z}^{obs}}=d_i^{\mathbf{z}}] \sum_{d=0}^{n-1} \frac{g_{id}h_{id}}{f_{id}} \right)^{-1}
     \frac{g_{id_i^{\mathbf{z}}} h_{id_i^{\mathbf{z}}} }{f_{id_i^{\mathbf{z}}}} ( 2z_i  -1)
\end{align*}
Plugging in
\[
    g_{id} = \frac{\Pr[z_i^{obs}=1,d_i^{\mathbf{z}^{obs}} = d]}{\sigma(\mathbf{z'})} 
  \text{ and }
    h_{id} = \frac{\Pr[z_i^{obs}=0,d_i^{\mathbf{z}^{obs}} = d]}{\sigma(\mathbf{z}'')},
\]
where $\mathbf{z}'$ and $\mathbf{z}''$ are allocations where the treated degree for unit $i$ is $d$ and $z'_i=1$ and $z''_i=0$, yields the estimate in Eq.~\ref{eq:mivlue_sania} in Theorem~\ref{thm:mivlueSANIA}.
\end{proof}

The first part of Theorem~\ref{thm:mivlueSUTVA} is simply an application of Theorem~\ref{thm:mivlueSANIA}.

\subsection{Constant Prior Case for Vertex Transitive Graphs}\label{app:const_prior}

The constant prior case, where almost surely the parameters are constant across units, does not provide immediate simplifications or closed forms for the MIV LUE for general graphs and designs. 
At this point, only the vertex transitive case has been fully explored, see Section~\ref{ssub:general_prior}, in which case a very simple form for the MIV LUE arises.

\begin{proof}[Proof of Proposition~\ref{prop:vertex_trans}]
To begin, we will prove the result in the special case that the design $p$ is supported on a single orbit of the automorphism group acting on the set of allocations. 
Specifically, the orbit of an allocation $\mathbf{z}$ is the set of allocations $\tau(\mathbf{z})$ where $\tau$ is an automorphism of the group $g$ which acts on allocations such that $\tau(\mathbf{z})_i= z_{\tau^{-1}(i)}$.
Importantly, the set of orbits partitions the space of allocations \citep{godsil2013algebraic}.

The first step in this proof is to show that estimator is unbiased.
In this special case, since we assume that the design is symmetric with respect to automorphisms, $p$ must be uniform on the orbit. 
We recall the notation $n_{z,d}(\mathbf{z})$ is the number of units $i$ where $z_i=z$ and $d_i^{\mathbf{z}}=d$.
Note that $n_{z,d}(\mathbf{z})=n_{z,d}(\tau(\mathbf{z}))$ for any automorphism $\tau$ since reassigning treatments according to an automorphism $\tau$ must also reassign the treated degrees in according to $\tau$.
Hence, for now we will abbreviate $n_{z,d}(\mathbf{z})$ as simply $n_{z,d}$ since it is the same for all allocations in the support.

Before proceeding, we will briefly prove the following group Lemmas.
\begin{lemma}\label{lem:auto}
Suppose that $g$ is vertex transitive and let $T$ be the automorphism group for $g$.
For any $i,j,k\in[n]$, the number of automorphisms $\tau$ where $\tau(i)=j$ is equal to the number of automorphisms where $\tau(i)=k$.
\end{lemma}
\begin{proof}
For any $i,j\in[n]$, let $T_{ij}\subset T$ be the set of automorphisms which map $i$ to $j$.
For any $\tau_{jk}\in T_{jk}$ and $\tau_{ij}\in T_{ij}$ we have that $\tau_{jk}\circ \tau_{ij}\in T_{ik}$.
Additionally, by the fact that $\tau_{jk}$ has an inverse, if $\tau_{ij},\tau'_{ij}\in T_{ij}$ then  $\tau_{ij}\neq \tau'_{ij}$ implies $\tau_{jk}\circ \tau_{ij}\neq \tau_{jk}\circ \tau'_{ij}$ for any distinct $\tau_{ij},\tau'_{ij}\in T_{ij}$.
This implies that $|T_{ij}|\leq |T_{ik}|$ and reversing the roles of $j$ and $k$ above implies the result.
\end{proof}

\begin{lemma}
Under the conditions of Proposition~\ref{prop:vertex_trans} where the design $p$ is supported on a single orbit, $\Pr[z_i^{obs}=z, d_i^{\mathbf{z}^{obs}}=d]=n_{z,d}/n$.
\end{lemma}
\begin{proof}
Select a fixed allocation $\mathbf{z}\in \mathrm{supp}(\mathbf{z}^{obs})$.
We can sample from $p$ by selecting a random automorphism $\tau$ drawn uniformly from all automorphisms and then $\mathbf{z}^{obs}=\tau(\mathbf{z})\sim p$. 
We now have that $\Pr[z_i^{obs}=z, d_i^{\mathbf{z}^{obs}}=d]$ is the probability that we draw an allocation that maps $i$ to a unit $j$ where $d_j^{\mathbf{z}}=d$ and $z_j=z$.
The number of such units is exactly $n_{z,d}$ and by Lemma~\ref{lem:auto} the probability that the random automorphism $\tau$ maps $i$ to $j$ for any units $i$ and $j$ is exactly $|T_{ij}|/\sum_k |T_{ik}| = 1/n$.
\end{proof}

Now, recall from Proposition~\ref{prop:vertex_trans} Eq.~\eqref{eq:strat_naive}, the estimator is defined as
\[
w_i(\mathbf{z}) = \frac{C_{d_i^{\mathbf{z}}}}{\sum_{d=0}^{n-1} C_d} \frac{2z_i-1}{n_{z_i,d_i^{\mathbf{z}}}} 
\]
where \[
    C_d=\mathbf{I}\{n_{0,d}>0, n_{1,d}>0 \} \left( \frac{1}{n_{0,d}}+\frac{1}{n_{1,d}} \right)^{-1}\text{ and } C= \sum_{d=0}^{n-1} C_d.
\]
Altogether, we have that
\begin{align*}
&\Ex[w_i(\mathbf{z}^{obs})Y_i(\mathbf{z}^{obs})] \\
=& \sum_{d=0}^{n-1} \frac{C_d}{C} \left( \Pr[z^{obs}_i=1,d^{\mathbf{z}^{obs}}_i=d] \frac{1}{n_{1,d}}(\alpha_i+\beta_i+\Gamma_i(d))\right.\\
& \qquad\qquad\qquad\qquad \left.-\Pr[z^{obs}_i=0,d^{\mathbf{z}^{obs}}_i=d]\frac{1}{n_{0,d}}(\alpha_i+\Gamma_i(d))\right) \\ 
=&\sum_{d=0}^{n-1} \frac{C_d}{C} \left(   \frac{n_{1,d}}{n}\frac{1}{n_{1,d}}(\alpha_i+\beta_i+\Gamma_i(d))-\frac{n_{0,d}}{n}\frac{1}{n_{0,d}}(\alpha_i+\Gamma_i(d))\right) \\
=&\sum_{d=0}^{n-1} \frac{C_d}{C} \left(   \frac{1}{n}(\alpha_i+\beta_i+\Gamma_i(d))-\frac{1}{n}(\alpha_i+\Gamma_i(d))\right) \\
=& \sum_{d=0}^{n-1} \frac{C_d}{C} \frac{1}{n}\beta_i = \frac{1}{n}\beta_i.
\end{align*}
This implies that the estimator is unbiased. 

The remaining KKT conditions are that 
\begin{equation}
  \Sigma(\mathbf{z})\mathbf{w}(\mathbf{z})  -  \lambda_\alpha  - \mathrm{diag}(\mathbf{z}) \lambda_\beta  - \sum_{d=1}^{n-1} \mathrm{diag}(\mathbf{I}\{d^{\mathbf{z}} =d\})\lambda_{\Gamma d}=0,\label{eq:KKT2}
\end{equation}
for each $\mathbf{z}$ in the support
where we are free to choose each $\lambda$ term, but the $\lambda$ terms cannot depend on $\mathbf{z}$.
Since all the parameters are equal across units, let us define $\sigma_{\xi,\xi'}=\mathrm{Cov}(\xi,\xi')$ where $\xi,\xi'$ can be any of $\alpha$, $\beta$, $\Gamma(1), \dotsc, \Gamma(n-1)$.
We have \[
  \Sigma(\mathbf{z})_{ij} = \sigma_{\alpha,\alpha}+(z_i+z_j) \sigma_{\alpha,\beta}+z_i z_j \sigma_{\beta,\beta}
  +\sigma_{\Gamma(d_i^{\mathbf{z}}), \Gamma(d_j^{\mathbf{z}})} +
   z_i\sigma_{\beta, \Gamma(d_j^{\mathbf{z}})}+
   z_j \sigma_{\beta, \Gamma(d_i^{\mathbf{z}})}+
   \sigma_{\alpha, \Gamma(d_j^{\mathbf{z}})}+
   \sigma_{\alpha, \Gamma(d_i^{\mathbf{z}})}+\sigma_{\alpha, \Gamma(d_j^{\mathbf{z}})}
\]
Note that since $\sum_j w_j(\mathbf{z})=0$, we have that $(\Sigma(\mathbf{z})\mathbf{w}(\mathbf{z}))_i$ will have no terms from above that depend only on $z_i$ and $d_i^{\mathbf{z}}$, including $\sigma_{\alpha, \alpha}$, $z_i\sigma_{\alpha,\beta}$, $\sigma_{\alpha,\Gamma_i(d_i^{\mathbf{z}})}$, \ldots.
Let $w_{z,d}=w_i(\mathbf{z})$ for any unit $i$ where $d_i^{\mathbf{z}}=d$ and $z_i=z$ and note that $w_{z,d}n_{z,d}=(2z-1)C_d/C$.
Hence,
\begin{align*}
(\Sigma(\mathbf{z})\mathbf{w}(\mathbf{z}))_i =\sum_{d=0}^{n-1} w_{0,d}n_{0,d} &\left (\sigma_{\alpha,\alpha}+z_i \sigma_{\alpha,\beta}
  +\sigma_{\Gamma(d_i^{\mathbf{z}}), \Gamma(d)} +
   z_i\sigma_{\beta, \Gamma(d)}+
   \sigma_{\alpha, \Gamma(d)}+
   \sigma_{\alpha, \Gamma(d_i^{\mathbf{z}})}+\sigma_{\alpha, \Gamma(d)}) \right)\\
+w_{1,d}n_{1,d}&\left( \sigma_{\alpha,\alpha}+z_i \sigma_{\alpha,\beta}+\sigma_{\alpha,\beta}+z_i \sigma_{\beta,\beta}
  +\sigma_{\Gamma(d_i^{\mathbf{z}}), \Gamma(d)} +
   z_i\sigma_{\beta, \Gamma(d)}+
   \sigma_{\beta, \Gamma(d_i^{\mathbf{z}})}+\right. \\
   &\left. \sigma_{\alpha, \Gamma(d)}+
   \sigma_{\alpha, \Gamma(d_i^{\mathbf{z}})}+\sigma_{\alpha, \Gamma(d)}\right) \\
   =\sum_{d=0}^{n-1} \frac{C_d}{C} (\sigma_{\alpha,\beta}+&z_i\sigma_{\beta,\beta}+\sigma_{\beta, \Gamma(d_i^{\mathbf{z}})}) = \sigma_{\alpha,\beta}+z_i\sigma_{\beta,\beta}+\sigma_{\beta, \Gamma(d_i^{\mathbf{z}})}
\end{align*}
or, succinctly, 
\[
  \Sigma(\mathbf{z})\mathbf{w}(\mathbf{z}) =  \sigma_{\alpha,\beta}\bm{1}  +\sigma_{\beta,\beta}\mathbf{z} + \sum_{d=1}^{n-1}\sigma_{\beta,\Gamma(d)} \mathbf{I}\{d^{\mathbf{z}} =d\})
\]
Hence, setting $\lambda_{\alpha}=\sigma_{\alpha,\beta}\bm{1}$, $\lambda_{\beta}=\sigma_{\beta,\beta}\bm{1}$, and $\lambda_{\Gamma d}=\sigma_{\beta,\Gamma(d)}\bm{1}$ will ensure that Eq.~\eqref{eq:KKT2} is satisfied.

This proves the result if $p$ is supported on a single orbit. 
If $p$ is not supported on a single orbit then we can write $p$ as a mixture of designs supported on each orbit.
Since the estimator does not depend on the design, by conditioning we can verify that the estimator is unbiased for any such mixture.
Similarly, the analysis of the remaining KKT conditions follows {\em mutatis mutandis}.
\end{proof}

\section{MIVLUE under SANASIA}\label{app:mivlue_sanasia}
In this section we will sketch a derivation of the MIVLUE under SANASIA for a specific set of priors. 
We assume that the the priors are uncorrelated across units for the $\alpha_i$ and $\beta_i$ parameters and that the prior for the $\gamma_i$ parameters are uncorrelated with the priors on the $\alpha_i$ and $\beta_i$ parameters.
As in the previous sections we assume that the priors are all mean zero.

We will also assume, in the context of Proposition~\ref{prop:paramSANASIA}, that the graph $h$ has exactly one connected component so that the interference effects are equal for all units. 
The results are similar for $h$ with multiple connected components but the derivation requires additional notation and book keeping.

Recall that under SANASIA $Y_i(\mathbf{z})= \alpha_i + \beta_i z_i +\gamma d_i^{\mathbf{z}}$.
The mean square error for a LUE is 
\begin{align*}
    MSE(\hat{\beta}^w)  
    &= \sum_{\mathbf{z}} p(\mathbf{z}) \sum_{i=1}^n\sum_{j=1}^n w_i(\mathbf{z})w_j(\mathbf{z})(\alpha_i+\beta_i z_i +\gamma d_i^{\mathbf{z}})(\alpha_j+\beta_j z_j +\gamma d_j^{\mathbf{z}})-\beta^2.
\end{align*}
We will write the prior variance for $\alpha_i$  as $\sigma_{\alpha i}^2$ and $\beta_i$  as $\sigma_{\beta,i}^2$. 
We write $\mathrm{Var}(\gamma)=\sigma_\gamma^2$.
Hence, half of the integrated mean square error can be written as
\begin{align*}
\frac{1}{2}\sum_{\mathbf{z}} p(\mathbf{z}) \sum_{i=1}^n w_i(\mathbf{z})^2 \left(  \sigma_{\alpha i}^2+\sigma_{\beta i}^2 z_i +   d_i^{\mathbf{z}}\sigma_\gamma^2 \sum_j d_j^{\mathbf{z}} \right) -\frac{\sigma_\beta^2}{2n}.
\end{align*}
The Lagrangian is 
\begin{align*}
L(\mathbf{w},\bm{\lambda})=& \frac{1}{2} \sum_{\mathbf{z}} p(\mathbf{z}) \sum_{i=1}^n w_i(\mathbf{z})^2 \left(  \sigma_{\alpha i}^2+\sigma_{\beta i}^2 z_i +   d_i^{\mathbf{z}}\sigma_\gamma^2 \sum_j d_j^{\mathbf{z}} \right) \\
+ \lambda_\alpha^T &\sum_{\mathbf{z}} p(\mathbf{z}) \mathbf{w}(\mathbf{z}) \\
+ \lambda_\beta^T & \left( \sum_{\mathbf{z}} p(\mathbf{z}) \mathbf{w}(\mathbf{z}) z_i - \frac{1}{n} \bm{1} \right) \\
+ \lambda_{\gamma} & \sum_{\mathbf{z}} p(\mathbf{z}) \sum_{i} w_i(\mathbf{z}) d_i^{\mathbf{z}}.
\end{align*}.
The derivative of the Lagrangian with respect to $w_i(\mathbf{z})$ is 
\[
  \frac{\partial L}{\partial w_i(\mathbf{z})} = p(\mathbf{z}) \left(  2 w_i(\mathbf{z})( \sigma_{\alpha i}^2+\sigma_{\beta i}^2 z_i + \sigma_\gamma^2 d_i^{\mathbf{z}} \sum_j d_j^{\mathbf{z}}) + \lambda_{\alpha i }  + \lambda_{\beta i }  z_i + \lambda_\gamma d_i^{\mathbf{z}} \right)
\]
so that \[
  w_i(\mathbf{z}) = -\frac{\lambda_{\alpha i}+\lambda_{\beta i}z_i +\lambda_\gamma d_i^{\mathbf{z}} }{\sigma_i(\mathbf{z})}.
\]
where we define $\sigma_i(\mathbf{z})= \sigma_{\alpha i}^2+\sigma_{\beta i}^2 z_i + \sigma_\gamma^2 d_i^{\mathbf{z}} \sum_j d_j^{\mathbf{z}}$.

The derivative with respect to the $\lambda$ terms after plugging in the solution for $w_i(\mathbf{z})$ above is 
\begin{align*}
   \frac{\partial L}{\partial \lambda_{\alpha i}} &= -\sum_{\mathbf{z}} p(\mathbf{z}) \frac{\lambda_{\alpha i}+\lambda_{\beta i }z_i +  \lambda_{\gamma} d_i^{\mathbf{z}} }{ \sigma_i(\mathbf{z})}\\
   \frac{\partial L}{\partial \lambda_{\beta i}} &= -\sum_{\mathbf{z}}  p(\mathbf{z}) \frac{\lambda_{\alpha i}z_i+\lambda_{\beta i }z_i +  \lambda_{\gamma} d_i^{\mathbf{z}} z_i}{\sigma_i(\mathbf{z})}  - 1/n\\ 
   \frac{\partial L}{\partial \lambda_{\gamma}} &= -\sum_{\mathbf{z}} p(\mathbf{z}) \sum_i \frac{\lambda_{\alpha i} d_i^{\mathbf{z}} +\lambda_{\beta i }d_i^{\mathbf{z}} z_i + \lambda_{\gamma} {d_i^{\mathbf{z}}}^2 }{\sigma_i(\mathbf{z})} 
\end{align*}
Hence, \[
  \lambda_\gamma = C^{-1}\sum_{\mathbf{z}} p(\mathbf{z}) \sum_i \frac{\lambda_{\alpha i} d_i^{\mathbf{z}} +\lambda_{\beta i }d_i^{\mathbf{z}} z_i }{\sigma_i(\mathbf{z})}
\]
where $C=-\sum_{\mathbf{z}} p(\mathbf{z}) \sum_i \frac{ {d_i^{\mathbf{z}}}^2 }{\sigma_i(\mathbf{z})} $.

Let $a_i$ and $b_i$ be as in Eq.~\eqref{eq:ai_bi} and let $g_i = \sum_{\mathbf{z}} p(\mathbf{z}) \frac{d_i^{\mathbf{z}}}{\sigma_i(\mathbf{z})} $ and let $h_i = \sum_{\mathbf{z}} p(\mathbf{z}) \frac{d_i^{\mathbf{z}} z_i}{\sigma_i(\mathbf{z})}$.

Plugging in the solution for $\lambda_{\gamma}$ in terms of $\lambda_\alpha$ and $\lambda_\beta$, we can write the derivatives with respect to $\lambda_\alpha$ and $\lambda_\beta$ as
\[
 -\left(   \begin{pmatrix}
    \mathrm{diag}(a) & \mathrm{diag}(b) \\ \mathrm{diag}(b) &\mathrm{diag}(b)
  \end{pmatrix} 
  + C^{-1}
  \begin{pmatrix}
    g \\ h
  \end{pmatrix}
  \begin{pmatrix}
    g^T & h^T
  \end{pmatrix}   \right) \begin{pmatrix}
    \lambda_\alpha \\ \lambda_\beta
  \end{pmatrix}
\]

The inverse of the leftmost matrix is 
 \[
 A^{-1}=  \begin{pmatrix}
    \mathrm{diag}(a-b)^{-1} & -\mathrm{diag}(a-b)^{-1} \\ -\mathrm{diag}(a-b)^{-1} &  \mathrm{diag}(a \circ b^{-1}  \circ(a-b)^{-1})
  \end{pmatrix}
\]
where $\circ$ denotes the Hadamard product of vectors, $(a \circ b)_i=a_i b_i$.
Using the Woodbury matrix identity \citep{higham2002accuracy} with the above, the solution for $\lambda_\alpha$ and $\lambda_\beta$ is
\[
\begin{pmatrix}
    \lambda_{\alpha}\\ \lambda_\beta 
  \end{pmatrix}
   = \left(   A^{-1} +  Q^{-1} A^{-1} \begin{pmatrix}
    g \\ h
  \end{pmatrix}   \begin{pmatrix}
    g^T & h^T
  \end{pmatrix} A^{-1} \right) \begin{pmatrix}
   \bm{0} \\ \frac{1}{n}\bm{1} 
  \end{pmatrix}
\]
where $Q$ is a scalar equal to 
\[
  C+ g^T \mathrm{diag}(a-b)^{-1} g -2 g^T \mathrm{diag}(a-b)^{-1} h -h^T  \mathrm{diag}(a \circ b^{-1}  \circ(a-b)^{-1}) h .
\]
So 
\[
\begin{pmatrix}
    \lambda_{\alpha}\\ \lambda_\beta 
  \end{pmatrix}
   = \frac{1}{n}\begin{pmatrix}
     -\frac{1}{a-b} \\ \frac{a}{b(a-b)}
   \end{pmatrix} - \frac{R}{n} \begin{pmatrix}
     \frac{g-h}{a-b} \\ \frac{h a-g b}{b(a-b)}
   \end{pmatrix}
   \]
where $R=Q^{-1} (h^T \frac{a}{b(a-b)} - g^T \frac{1}{a-b})$
Plugging this into $\lambda_\gamma$ and simplifying gives that
$\lambda_\gamma = C^{-1}( g^T \lambda_\alpha+ h^T \lambda_\beta) = \frac{RQ}{nC}-\frac{R (Q-C)}{nC}=R/n $

So we have \[
  w_i(\mathbf{z}) = \frac{1}{n\sigma_i(\mathbf{z})}\left(  \frac{-b_i-Rb_i(g_i-h_i)}{b_i(a_i-b_i)}+ z_i \frac{a_i-R(h_ia_i-g_ib_i)}{b_i(a_i-b_i)} + R d_i^{\mathbf{z}}  \right).
\]
This is the formula that is used for the MIV LUEs based on SANASIA in the simulations section. 

\begin{theorem}\label{thm:mivlueSANASIA}
Suppose the potential outcomes satisfy SANASIA and that the interference effect is constant so that $Y_i(\mathbf{z})=\alpha_i +\beta_i z_i +\gamma d_i^{\mathbf{z}}$.
Suppose the prior distribution on the parameters satisfies $\mathrm{Cov}((\alpha_i,\beta_i),(\alpha_j,\beta_j))=0$ and that $\gamma$ is indepo and the prior on the parameters has no correlation between units. 
If any unbiased estimators exist, then the weights for the MIV LUE are given by
\begin{equation}
    w_i(\mathbf{z}) = \frac{1}{n\sigma_i(\mathbf{z})}\left(  \frac{-b_i-Rb_i(g_i-h_i)}{b_i(a_i-b_i)}+ z_i \frac{a_i-R(h_ia_i-g_ib_i)}{b_i(a_i-b_i)} + R d_i^{\mathbf{z}}  \right).
    \label{eq:mivlue_sanasia}
\end{equation}
where we define the vectors $a,b,g,h\in \Re^n$ as
\begin{align}
    a_i &= \sum_{\mathbf{z}} \frac{p(\mathbf{z})}{\sigma_i(\mathbf{z})} &b_i &= \sum_{\mathbf{z}} \frac{p(\mathbf{z}) z_i }{\sigma_i(\mathbf{z})},\\
    g_i &= \sum_{\mathbf{z}} p(\mathbf{z}) \frac{d_i^{\mathbf{z}}}{\sigma_i(\mathbf{z})}  & h_i &= \sum_{\mathbf{z}} p(\mathbf{z}) \frac{d_i^{\mathbf{z}} z_i}{\sigma_i(\mathbf{z})},
\end{align}
The scalar $R$ is defined in terms of 
\begin{align*}
C&=-\sum_{\mathbf{z}} p(\mathbf{z}) \sum_i \frac{ {d_i^{\mathbf{z}}}^2 }{\sigma_i(\mathbf{z})}\in \Re,\text{ and }\\
Q&=C+ g^T \mathrm{diag}(a-b)^{-1} g -2 g^T \mathrm{diag}(a-b)^{-1} h -h^T  \mathrm{diag}(a \circ b^{-1}  \circ(a-b)^{-1})h \in \Re,\\
R&=Q^{-1} \left(h^T (a \circ b^{-1} \circ (a-b)^{-1}) - g^T (a-b)^{-1} \right)\in \Re.
\end{align*}
\end{theorem}

\end{document}